\documentclass[preprint]{imsart}
\RequirePackage[OT1]{fontenc}
\RequirePackage{amsthm,amsmath}
\RequirePackage[numbers]{natbib}
\RequirePackage[colorlinks,citecolor=blue,urlcolor=blue]{hyperref}

\issue{}
\firstpage{}
\lastpage{}

\startlocaldefs
\usepackage[tikz]{bclogo}
\usepackage{algorithm2e}

\usepackage{amssymb, amsmath, natbib, amsthm, graphicx, euscript, mathrsfs }
\def\sec#1{\underline{\textbf{#1}}}
\def\ttt#1{\textbf{#1}}
\def\R{I\!\!R}
\def\I{I\!\!I}
\renewcommand{\u}{u^{\circ}}

\newcommand{\A}{\mathscr{A}}

\newcommand{\eqd}{\buildrel d \over}

\newcommand{\PP}{{\mathbb P}}

\newcommand{\N}{{\mathbb N}}

\newcommand{\T}{\mathcal{T}}
\newcommand{\E}{{\mathbb E}}

\newcommand{\ii}{\mathrm{i}}
\newcommand{\eps}{\varepsilon}
\newcommand{\Ree}{\operatorname{Re}}
\newcommand{\Imm}{\operatorname{Im}}
\newcommand{\F}{\mathcal{F}}
\theoremstyle{plain}
\newtheorem{thm}{Theorem}[section]
\newtheorem{lem}[thm]{Lemma}

\newtheorem{rem}[thm]{Remark}
\newtheorem{prop}[thm]{Proposition}
\newcommand{\argmin}{\operatornamewithlimits{arg\,min}}

\newcommand{\B}{\widetilde{B}}
\newcommand{\W}{\mathcal{W}}
\newcommand{\C}{{\mathbb C}}
\newcommand{\vsp}{\vspace{0.5cm}}
\newcommand{\K}{\mathcal{K}}
\newcommand\erf{\operatorname{erf}}
\newcommand{\balpha}{\bar{\alpha}}
\newcommand{\p}{\varphi}

\begin{document}
\begin{frontmatter}
\title{Statistical inference for exponential functionals of L\'evy processes\thanksref{T1}}
\runtitle{Statistical inference for exponential functionals}

\begin{aug}
\author{\fnms{Denis} \snm{Belomestny}\ead[label=e1]{denis.belomestny@uni-due.de}}
\address{University of Duisburg-Essen\\
Thea-Leymann-Str. 9, 45127 Essen,  Germany\\
\printead{e1}}

\author{\fnms{Vladimir} \snm{Panov}\ead[label=e2]{vpanov@hse.ru}}
\address{National research university Higher School of Economics \\ Shabolovka, 26, Moscow,  119049 Russia.\\
\printead{e2}}

\thankstext{T1}{This research was partially supported by the Deutsche
Forschungsgemeinschaft through the SFB 823 ``Statistical modelling of nonlinear dynamic processes'' and  by
Laboratory for Structural Methods of Data Analysis in Predictive Modeling, MIPT, RF government grant, ag. 11.G34.31.0073.}
\runauthor{D.Belomestny and V.Panov}

\end{aug}

\begin{abstract}
In this paper, we consider the exponential functional \(A_{\infty}=\int_0^\infty e^{-\xi_s}ds\) of a L{\'e}vy process \(\xi_s\) and aim to estimate the characteristics of \(\xi_{s}\) from the distribution of \(A_{\infty}\). 
We present a new approach, which allows to statistically infer on the L{\'e}vy triplet of \(\xi_{t}\), and  study the theoretical properties of the proposed estimators. The suggested  algorithms are illustrated with numerical simulations. 
\end{abstract}

\tableofcontents
\end{frontmatter}

\section{Introduction}
For a L{\'e}vy process \(\xi=\left(\xi_{t}\right)_{t \geq 0}\), the exponential functional of \(\xi\) is defined by 
\begin{eqnarray*}
	A_{t} = \int_{0}^{t} e^{-\xi_{s}} \; ds,
\end{eqnarray*}
where \(t \in (0,\infty)\).
The main object of this research is the terminal value \begin{eqnarray}
\label{Ainfty}
A_{\infty}:=\lim_{t\to\infty} A_{t}= \int_{0}^{\infty} e^{-\xi_{s}} \; ds,
\end{eqnarray}  
which often (and everywhere in this paper) is called also an exponential functional of \(\xi\). 
The integral \(A_{\infty}\) naturally arises in a wide variety of financial applications as an invariant distribution of the process
\begin{eqnarray}
\label{GOU}
V_{t} = e^{-\xi_{t}} \left(  
	V_{0} + \int_{0}^{t} e^{\xi_{s-}} d s
\right),
\end{eqnarray}
see Carmona, Petit, Yor \cite{CPY2}. For instance, the process \eqref{GOU}  determines the volatility process in the COGARCH (COntinious Generalized AutoRegressive Conditionally Heteroscedastic) model  introduced by Kl{\"u}ppelberg et al. \cite{cogarch}. Note that \(V_{t}\) is in fact a partial case of the generalized Ornstein-Uhlenbeck (GOU) process. A comprehensive study of the GOU model is given in the dissertation by Behme \cite{Behme}.

 \(A_{\infty}\) appears in finance also  in other contexts, for instance, in pricing  of Asian options, see the monograph by Yor \cite{Yor} and the references given by Carmona, Petit, Yor \cite{CPY2}.   As for other fields of applications, \(A_{\infty}\) plays a crucial role in studying the carousel systems (see Litvak and Adan \cite{LitvakAdan}, Litvak and Zwet \cite{LitvakZwet}),  self-similar fragmentations (see Bertoin and Yor \cite{BertoinYor}), and information transmission problems (especially TCP/IP protocol, see Guillemin, Robert and Zwart \cite{Guill}).   For the detailed discussion of the physical interpretations, we refer to Comtet, Monthus and Yor \cite{CMY} and the dissertation by  Monthus \cite{Monthus}.

Denote the  L{\'e}vy triplet of the process \(\xi_{t}\) by \((c, \sigma, \nu)\), i.e.,
\begin{eqnarray}
\label{xi}
\xi_{t} =
	c t + \sigma W_{t} + \T_{t},
\end{eqnarray}
where \(\T_{t}\) is a pure jump process with L{\'e}vy measure \(\nu\).    The finiteness condition  stands that the integral \(A_{\infty}\) is finite if and only if \(\xi_{t} \to +\infty\) as \(t \to +\infty\), see  Maulik and Zwart \cite{MZ} for the proof and Erickson and Maller \cite{ErMaller} for some extensions of this result. Therefore,  the integral \(A_{\infty}\) is finite if the process \(\xi_{t}\) is  any non-degenerated subordinator, i.e., any  non-decreasing L{\'e}vy process, or, equivalently, any non-negative L{\'e}vy process. Nevertheless, the finiteness condition is fulfilled for other processes also, e.g., for \(\xi_{t} = -N_{t} + 2 \lambda t\), where \(N_{t}\) is a Poisson process with intensity \(\lambda\).

In this paper, we mainly focus on the case when  \(\xi\) is a subordinator with finite L{\'e}vy measure. In terms of the L{\'e}vy triplet, this means that \(c>0\), \(\sigma=0\), \(\nu(\R_{-})=0\) and moreover \(a:=\nu(\R_{+})<\infty\).  Suppose that the process \eqref{GOU} is observed in the time points \(0=t_{0} < t_{1} < .. .< t_{n}\). Taking into account that the process \(V_{t}\) is a Markov process, and assuming that \(V_0\) has an invariant distribution determined by \(A_{\infty}\), we conclude that \(V_{t_{0}},  ... , V_{t_{n}}\)  have also the distribution of  \(A_{\infty}\). 
 The main goal of this research is to  statistically infer on the L{\'e}vy triplet of \(\xi\) from the observations  \(V_{t_{0}},  ... , V_{t_{n}}\).   More precisely, we will pursue the following two aims: (1) to estimate the drift term \(c\)  and the parameter \(a\); (2) to estimate the L{\'e}vy measure \(\nu\). 
 
To the best of our knowledge, the statistical inference for exponential functionals of L{\'e}vy processes has not been previously considered in  the literature. However, some distributional properties of the exponential functionals are well-known, e.g., the integro - differential equation by Carmona, Petit, Yor \cite{CPY}. For the overview of  theoretical results, we refer to the survey by Bertoin and Yor \cite{BertoinYor}. One distribution property, the recursive formula for the moments of \(A_{\infty}\), gives rise to the approach presented in our paper. 
This result stands that   
\begin{eqnarray}
\label{momenty}
	\E \left[
		A_{\infty}^{s-1}
	\right]
=
\frac{
	\psi(s)
}{
	s
}\;
\E \left[
		A_{\infty}^{s}
	\right]
,
\end{eqnarray}
where \(\psi(s)\) is a Laplace exponent of the process \(\xi\), i.e., 
\(
	\psi(s) := - \log \E \left[
		e^{ - s \xi_{1}}
	\right],
\)
and complex \(s\) is taken from the area
 \begin{eqnarray}
 \label{Ups}
 \Upsilon:= \Bigl\{s \in \C: \;  0<\Ree(s)  < \theta \Bigr\} \quad \mbox{with}\quad  \theta:=\sup\left\{z \geq 0: \E[ e^{-z \xi_{1}}] \leq 1\right\}.
 \end{eqnarray}
The recursive  formula \eqref{momenty} firstly appear  for real \(s\) in the paper by Maulik and Zwart \cite{MZ}. The complete proof for complex \(s\) was given recently  by Kuznetsov, Pardo and Savov \cite{Kuznets}. If \(\xi_{t}\) is a subordinator, what is the case under our setup, the parameter \(\theta\) is equal to infinity.

The idea of the procedure for solving the first task (estimation of \(a\) and \(c\)) is to infer on the parameters of the process \(\xi\)  from its Laplace exponent.  First, making use of \eqref{momenty}, we estimate the Laplace exponent \(\psi(s)\) at the points \(s=u+\ii v \in \Upsilon\), where \(u\) is fixed and \(v\) varies on the equidistant grid between \(\eps V_{n}\) and \(V_{n}\) (with \(\eps>0\) and \(V_{n} \to \infty\) as \(n \to \infty\)) Afterwards, we take into account that 
\begin{eqnarray}
	\label{estproc1}
	\psi(u+\ii v) &=& a+ c \left( u+\ii v \right) -  \F_{\bar\nu}(-v) , \qquad u,v \in \R,
\end{eqnarray}
where  \(\bar\nu(dx):= e^{-u x} \nu(dx)\), and \(\F_{\bar\nu} (v)\) stands for the Fourier transform of the measure \(\bar\nu\), i.e., \( \F_{\bar\nu} (v) := \int_{\R_{+}} e^{\ii v x} \bar\nu(dx).\) It is worth mentioning that \(\F_{\bar\nu} (v) \to 0\) as \(v \to \infty\), and therefore taking the real and imaginary parts of the left and right hand sides of \eqref{estproc1}, we are able to consequently estimate the parameters \(c\) and \(a\).

With no doubt, the second aim (complete recovering of the L{\'e}vy measure) is the most challenging task. Since the estimates of the parameters \(c\) and \(a\) are already obtained, we can estimate by \eqref{estproc1} the Fourier transform \(\F_{\bar\nu} (v)\) for \(v\) taken from the equidistant grid \([-V_{n}, V_{n}]\).   The last step of this procedure, estimation of the L{\'e}vy measure \(\nu\),  is based on the inverse Fourier transform formula, and  reveals the main reason for using the complex numbers in our approach. In fact, one can estimate the function \(\psi(\cdot)\) in real points and then estimate the Laplace transform of the measure  \(\nu\) by the regression arguments. In this case, estimation of \(\nu\) demands the inverse Laplace transform, which is given by Bromwich integral and therefore is in fact much more involved in comparison with the inverse Fourier transform. 

The paper is organized as follows. In the next section, we introduce the assumptions on this model and give some examples. We formulate the algorithms for estimation the Laplace exponent \(\psi(s)\) (Section 3.1),  the parameters \(a\) and \(c\) (Section 3.2), and the L{\'e}vy measure \(\nu\) (Section 3.3). 
Next, we provide some numerical examples in Section 4 and analyze the convergence rates of the proposed algorithms  in Section 5. Appendix contains some related results and additional proofs.

\section{Assumptions on the model}
\subsection{Subordinators}
\label{subo}
 In this article, we restrict our attention to the case when the following set of assumptions is fulfilled:
 \begin{description}
    \item[(A1)] \qquad\qquad \qquad
    $$
    \left\{
    \begin{aligned}
    c\geq 0, &\qquad \sigma=0, \\
    \nu(\R_{-}) = 0, &\qquad a:=  \nu(\R_{+}) < \infty.\\
    \end{aligned}
    \right.
    $$ 
    \end{description}
   This set in particularly yields that the process \(\xi\) has finite variation, i.e.,
   \begin{eqnarray}
\label{nu}
	\int_{\R_{+}} \left(
		x \wedge 1
	\right) \nu(dx) < \infty,
\end{eqnarray}
and therefore \(\xi\) is a non-decreasing L{\'e}vy processes, i.e., a subordinator. The detailed discussion of the subordination theory as well as various examples of such processes (Gamma, Poisson,  tempered stable, inverse Gaussian, Meixner processes, etc.), are given in  \cite{BNS}, \cite{Bertoin}, \cite{ContTankov},  \cite{Sato}, \cite{Schoutens}.

Note that in the case of subordinators, the truncation function in the L{\'e}vy-Khinchine formula can be omitted, and therefore  the characteristic exponent of \(\xi\) is equal to 
\begin{eqnarray}
\label{phis}
	\psi_{e}(s) =  \log \E \left[
		e^{\ii s \xi_{1}}
	\right] 
	=
	 \ii c s  
	 + \int_{0}^{\infty} 
	\left(
		e^{\ii s x}  - 1
	\right) \nu(dx).
\end{eqnarray}      
Later on, we use  a Laplace exponent of \(\xi\), which is defined by \[\psi(s) := - \log \E \left[
		e^{-s \xi_{1}}
	\right] =  -\psi_{e} \left(\ii s \right),\] and under the assumption (A1)  is equal to
\begin{eqnarray}
\label{phis2}
	\psi(s) &=& 
	c s 
	+ \int_{0}^{\infty}
	\left(
		1 -  e^{- s x} 
	\right) \nu(dx) \\
\label{phis3}
&=& 
	c s 
	+ s \int_{0}^{\infty} 
	  e^{- s x} \nu(x, +\infty) dx.
\end{eqnarray}
Some examples  can be found in Section \ref{secsim}.  In the sequel, we use the fact that the function \(\psi(\cdot)\) is bounded from above on the set \(\Upsilon\) by 
\begin{eqnarray*}
	|\psi(s)| &\leq& c |s| 
	+ 
	\int_{0}^{\infty} 
	\left(
		1 + e^{-  \Ree(s) x} 
	\right) \nu(dx)  
	\leq c \sqrt{\theta^{2} + \Imm^{2}(s)} 
	+ 
	2 a,
\end{eqnarray*}
and hence the asymptotic behavior of the function \(\psi(s)\) is given by
\begin{eqnarray}
\label{upp}
\left|
	\psi(s)
\right| 
=
O\Bigl(\Imm (s)\Bigr), \qquad \Imm(s) \to +\infty.
\end{eqnarray}

\subsection{Further assumptions on $\nu$}
\label{further}
First, we assume  the following asymptotic behavior of the Mellin transform of the integral \(A_{\infty}\):
\begin{description}
    \item[(A2)]  \(\hspace{3cm}  \left| \E \left[ A_{\infty}^{u^{\circ}+ \ii v} \right] \right| \asymp \exp\{ - \gamma |v|\}, \qquad \mbox{as }\;  |v|  \to \infty\)
\end{description}
with some \(\gamma>0\) and \(u^{\circ}>0\). 

Second, we introduce an assumption on 
the measure \(\bar\nu(dx) := e^{-u^{\circ}x} \nu(dx)\) : 
\begin{description}
    \item[(A3)]  \(\hspace{4cm} \left\| \bar\nu^{(r)} \right\|_{L^{\infty}(\R)}\leq C \) 
\end{description}
for some positive \(r\) and \(C\).

It is worth noting that there is an (indirect) relation between the assumptions (A2) and (A3). In fact, joint consideration of \eqref{momenty} and \eqref{estproc1} yields that 
\begin{eqnarray*}
	\F_{\bar\nu} (-v) &=& a+ c \left( \u+\ii v \right) -  (\u+\ii v) \frac{m(\u+ \ii v) }{m((\u+1)+ \ii v)}, \qquad u,v \in \R,
\end{eqnarray*}
where \(m(s):=\E \left[A_{\infty}^{s-1}\right]\) is the Mellin transform. 

\sec{Example 1.} For instance, the set of assumptions (A1) - (A3)  fulfills for  the class of L{\'e}vy processes with \(c=0\) and L{\'e}vy  density in the form
\[
\nu(x) = I_{x>0} \sum_{j=1}^{M} \sum_{k=1}^{{m_{j}}} \alpha_{jk} x^{k-1} e^{-\rho_{j}x}
\]
with \(M, m_{j} \in \N\), \(\rho_{j}>0\), \(\alpha_{jk}>0\).  In fact, assumption (A1) and (A3) obviously hold; assumption (A2) is checked in  \cite{Kuznets2} for any positive \(u^{\circ}\)  (p. 658, the proof of Theorem 1). 

\sec{Example 2.} Next, we provide an example of the L{\'e}vy process which doesn't possess the property (A2). Consider a subordinator \(\T\) with drift \(c>0\) and the L{\'e}vy density 
\begin{eqnarray*}
\nu(x)= ab \exp\{-bx\}\: I\{x>0\}, \quad a, b >0,
\end{eqnarray*}
which we describe in details in Section~\ref{secsim}. The exponential functional of this process has a density
\begin{eqnarray*}
	k(x) = C_{1} x^{b} (1-c x)^{(a/c) -1 } \: I\{0< x<1/c\}
\end{eqnarray*}
with some \(C_{1}>0\),  see \cite{CPY}.  In other words, the exponential functional has a distribution \(B(\alpha+1,\beta+1)/c\), where \(B\) stands for Beta distribution with parameters \(\alpha=b\) and \(\beta=a/c-1\). The Mellin transform of the function \(k(x)\) in the half -plane \(\Ree(s)>-\alpha\) is given by 
\begin{eqnarray*}
	m(s) = C_{2} (\alpha, \beta) \:c^{s}\: \frac{ \Gamma(\alpha+s)}{\Gamma(\alpha+\beta+1 + s)},
\end{eqnarray*}
where \(C_{2}(\alpha, \beta)>0\), see Table~1 from \cite{Eps}. Taking into account the following asymptotical behavior of the Gamma function 
\begin{eqnarray*}
 	\left| \Gamma(u+\ii v) \right|= \exp\left\{
		-\frac{\pi}{2} v +\left(u - \frac{1}{2} \right) \ln v + O(1) 
	\right\}, \qquad v \to \infty, 
\end{eqnarray*}
see (9) from \cite{Kuznets2}, we conclude that the exponential functional of the process \(\T\) has a polynomial decay of the Mellin transform. More precisely, \( |m(s)|  \asymp C_{3} v^{-a/c -1}\)  with some \(C_{3}>0\).


\section{Estimation of the L{\'e}vy triplet}
In the sequel, suppose that the process \eqref{GOU} is observed at the time points \(0=t_{0} < t_{1} < .. .< t_{n}\). Assuming that \(V_{0}\) has a stationary  invariant distribution, we get that the values \(A_{\infty, k}:=V_{t_{k}}, \; k=1..n\) have the distribution of the integral \(A_{\infty}\).
\subsection{Estimation of the Laplace exponent}
\label{Laplace} 
The first step of the estimation procedure is to construct the estimate of the function \(\psi(s)\) in the complex points \(s=u+\ii v\), where \(u\) is fixed and \(v\) varies. The reason for such choice of \(s\)  is clear from the further steps of the algorithm.   

The estimator of \(\psi(s)\) is based on  a recursive formula for the \(s-\)th (complex) moment of \(A_{\infty}\):
\begin{eqnarray}
\label{moment} 
	\E \left[
		A_{\infty}^{s-1}
	\right]
=
\frac{
	\psi(s)
}{
	s
}\;
\E \left[ 
		A_{\infty}^{s}
	\right],
\end{eqnarray}
In \cite{CPY}, this formula is proved  for real positive \(s\) such that \(\psi(s)>0\) and \(\E \left[
		A_{\infty}^{s}
	\right] < \infty\). The case of infinite mathematical expectations is carefully discussed in \cite{MZ}.  
	
	The case of complex \(s\) is considered in \cite{Kuznets}, where one can find also some generalizations of the formula \eqref{moment} for  integrals with respect to the Brownian motion with drift. In particular, applying Theorem 2 from \cite{Kuznets}, we get that  \eqref{moment} holds for any \(s \in \Upsilon\). In the situation when  \(\xi_{t}\) is a subordinator, the set \(\Upsilon\) coincides with the positive half-plane (equivalently, the parameter \(\theta\)  is equal to infinity), because it follows from \eqref{phis3} that
\[
	\E\left[
		e^{-s \xi_{1}} 
	\right]
	= - \psi (-s) = - cs - s \int_{\R_{+}} e^{sx} \nu (x, +\infty) dx < 0, \quad s>0.
\]

Motivated by \eqref{moment}, we now present the the first two steps in the estimation procedure. Let the values \(\alpha_{1}, ..., \alpha_{M}\) compose the equidistant grid with the step \(\Delta>0\) on the set \([\eps, 1]\), where  \(\eps>0\)  and the sequence \(V_{n}\) tends to infinity.
First, we 
estimate \(A_{\infty}^{s}\)  for \(s=u+\ii \alpha_{m} V_{n}, \; m=1..M\) and \(s=(u-1)+\ii v_{m}, \; m=1..M\)  where \(u:= u^{\circ} \in \left(-1,\theta\right)\) satisfies the assumptions (A2) and (A3).    Theoretical studies (see Section~\ref{theory}) show that the optimal choice is \(V_{n}= \kappa \log(n)\) with \(\kappa<1/(2\gamma)\), provided that the Assumptions (A1)-(A3) hold.  The estimator of \(A_{\infty}^{s}\) is defined by
\begin{eqnarray}
\label{step1}
	\widehat	\E_{n} \left[
		A_{\infty}^{s}
	\right]
	=
	\frac{1}{n}
	\sum_{k=1}^{n} A_{\infty, k}^{s}.
\end{eqnarray}
Next, we define an estimate of \(\psi(\cdot)\) at the points \(\left( u+ \ii \alpha_{m} V_{n} \right)\)  by
\begin{eqnarray}
\label{step2}
	\hat{\psi}_{n}(u+ \ii \alpha_{m} V_{n})
	=
	\left(
		u + \ii \alpha_{m} V_{n}
	\right)\;
\frac{
	\widehat\E _{n}\left[
		A_{\infty}^{(u-1) + \ii \alpha_{m} V_{n}}
	\right]
}{
	\widehat\E_{n} \left[
		A_{\infty}^{u + \ii \alpha_{m} V_{n}}
	\right]
}, \qquad m=1..M.
\end{eqnarray}
The performance of this estimator is later shown in Section~\ref{secsim}, see in particularly  Figures~\ref{plot1} and \ref{fig3}. The quality of \(\hat\psi_{n} (\cdot) \)  is theoretically studied in Theorem~\ref{thmphi}, which stands that under the assumptions (A1) and (A2) and the following condition on \(V_{n}\)
\begin{eqnarray*}
     \Lambda_{n} := V_{n} \exp\left\{c V_{n}\right\} \sqrt{\log V_{n}} = o \left(\sqrt{\frac{n}{\log(n)}} \right), \; n \to \infty,
\end{eqnarray*} 
it holds for \(n\) large enough
\begin{eqnarray}
\label{hatphi}
\PP\Biggl\{
\sup_{v \in [\eps V_{n}, V_{n}]}
\left|
	\hat\psi_{n}(u+\ii v) - \psi(u + \ii v)
\right|
\leq \beta \: \Lambda_{n}\sqrt{\frac{\log(n)}{n}} 
\Biggr\}  > 1 - \alpha n^{-1-\delta},
\end{eqnarray}
with some positive \(\alpha\), \(\beta\) and \(\delta\).

 \subsection{Estimation of $a$ and $c$}
In Section~\ref{subo}, we present the representation \eqref{phis2} for the Laplace exponent of the process \(\xi\). Substituting now the complex argument \(z=u+\ii v\), we get
\begin{eqnarray}
	\nonumber
	\psi(u+\ii v) &=& c \left( u+\ii v \right) - \int_{\R_{+}} e^{-\ii v x}\bar\nu(dx)  + \int_{\R_{+}} \nu(dx) \\
	\label{estproc}
	&=&c \left( u+\ii v \right) -  \F_{\bar\nu}(-v) +a , \qquad u,v \in \R,
\end{eqnarray}
where \(a:=\int_{\R_{+}} \nu(dx)\) and  \(\bar\nu(dx):= e^{-u x} \nu(dx)\). The general idea of the  procedure described below is to estimate the Laplace exponent \(\psi(\cdot)\) at the points \(s=u+\ii v\), where \(u\) is fixed at \(v\) varies (see Section~\ref{Laplace}), and afterwards to use  \eqref{estproc} for consequent estimation of the parameters. 

Taking imaginary and real of both hand sides in \eqref{estproc}, we get 
\begin{eqnarray}
\label{Im}
\Imm \psi(u+\ii v) &=& c v - \Imm \F_{\bar\nu} (-v), \\
\label{Re}
\Ree \psi(u+\ii v) &=& c u  - \Ree \F_{\bar\nu} (-v) + a.
\end{eqnarray}
By Riemann - Lebesque lemma,  \(\F_{\bar\nu} (-v) \to  0\) as \(v \to +\infty\), see, e.g., \cite{Kawata}; note that the rates of this convergence are assumed in (A3). Therefore, looking at \eqref{Im}, we conclude that \( \Imm \psi(u+\ii v)\)  is a (asymptotically) linear in \(v\) function, and   the parameter \(c\) can be interpreted a slope parameter. Next, from \eqref{Re}, it follows that  \( \Ree \psi(u+\ii v)\)  tends to \(( c u +a )\) as \(v \to +\infty\). These observations lead to the following optimization problems
\begin{eqnarray}
\label{optim}
\tilde{c}_{n} &:=&
\argmin_{c} \int_{\R_{+}} w_{n}(v) 
	\Bigl(
		\Imm \hat\psi_{n} (u+\ii v) 
		- 
		c v 
	\Bigr)^{2} dv \\
	\label{optim2}
\tilde{a}_{n} &:=&	
	\argmin_{a}\int_{\R_{+}} w_{n}(v) 
	\Bigl(
		\Ree \hat\psi_{n} (u+\ii v) 
		- 
		\tilde{c}_{n} u 
		-
		a
	\Bigr)^{2} 
	dv,
\end{eqnarray}
where the weighting function is chosen in the form \( w_{n}(v) =w(v/V_{n})/V_{n}\) with an integrable non-negative  function \(w (\cdot)\) supported on \([\eps,1]\).  Under this choice of \(V_{n}\), we can rewrite \eqref{optim} as follows:
\begin{eqnarray*}
\tilde{c}_{n} &:=&
\argmin_{c} \int_{\eps}^{1} w(\alpha) 
	\Bigl(
		\Imm \hat\psi_{n} (u+\ii \alpha V_{n}) 
		- 
		c \alpha V_{n}
	\Bigr)^{2} dv \\
\tilde{a}_{n} &:=&	
	\argmin_{a}\int_{\eps}^{1} w(\alpha) 
	\Bigl(
		\Ree \hat\psi_{n} (u+\ii \alpha V_{n}) 
		- 
		\tilde{c}_{n} u 
		-
		a
	\Bigr)^{2} 
	dv,
\end{eqnarray*}

In practice, we first get the estimates of the Laplace exponent at the points \(s=u+\ii \alpha_{m} V_{n}\) (see above) and  define an estimate of the parameter \(c\) by 
\begin{eqnarray}
\label{hatcopt}
	\hat{c}_{n} &:=& \argmin_{c} \sum_{m=1}^{M}  w(\alpha_{m}) 
	\Bigl(
		\Imm \hat\psi_{n} (u+\ii \alpha_{m} V_{n}) 
		- 
		c \alpha_{m} V_{n}
	\Bigr)^{2}\\
\label{hatc}
	&=& 
	 \frac{
	\sum_{m=1}^{M} 
		  w(\alpha_{m})  \alpha_{m} \: \Imm \hat\psi_{n} (u+\ii \alpha_{m} V_{n}) 
}
{	
	V_{n} \: \cdot \: \sum_{m=1}^{M} 
		  w(\alpha_{m})  \alpha_{m}^{2} 
}.
\end{eqnarray}
Afterwards, we estimate the parameter \(a\) by 
\begin{eqnarray}
\label{hataopt}
	\hat{a}_{n} &:=&
	\argmin_{a} \sum_{m=1}^{M}   w(\alpha_{m})
	\Bigl(
		\Ree \hat\psi_{n} (u+\ii \alpha_{m} V_{n}) 
		- 
		\hat{c}_{n} u - a
	\Bigr)^{2}\\
\label{hata}
	&=&
	\frac{
		\sum_{m=1}^{M} w (\alpha_{m}) 	\Ree \hat\psi_{n}(u+\ii \alpha_{m} V_{n}) 
	}
	{
		\sum_{m=1}^{M}
		w (\alpha_{m})
	}
		-\hat{c}_{n} u.
\end{eqnarray}

We show empirical and theoretical properties of the estimators \(\hat{a}_{n}\) and \(\hat{c}_{n}\) below,  see Figure~\ref{plot2} and Theorem~\ref{thm3}. Similar to \eqref{hatphi}, we prove that under the choice \(V_{n}=\kappa \log(n)\) with \(\kappa<1/(2\gamma)\), it holds 
\begin{eqnarray*}
\PP\Biggl\{
\left|
	\tilde{c}_{n} - c 
\right|
\leq \zeta_{1} \log^{-(r+2)}(n)
\Biggr\}  &>& 1 - \alpha n^{-1-\delta}, \qquad 
\mbox{and} \\
\PP\Biggl\{
\left|
	\tilde{a}_{n} -a 
\right|
\leq \zeta_{2}\log^{-(r+1)}(n)
\Biggr\}  &>& 1 - \alpha n^{-1-\delta}
,
\end{eqnarray*}
with \(\zeta_{1}, \zeta_{2}>0\), and \(\alpha, \delta\) introduced above. Constants \(\gamma\) and \(r\) involved in this statement are comming from  assumptions (A2) and (A3) resp. Moreover, we prove Theorem~\ref{lower}, which stands that this rate for \(c\) is optimal one in the class \(\A\) of the  models satisfying the assumptions (A1) - (A3). More precisely, we show that 
\[
\varliminf_{n \to \infty} 
\inf_{\tilde{c}_{n}^{*}} \sup_{\A}  
	\PP\Biggl\{
\left|
	\tilde{c}_{n}^{*} - c 
\right|
\geq \zeta_{3} \log^{-(r+2)}(n)
\Biggr\}  > 0,
\]
where \(\zeta_{3}<\zeta_{1}\) is some positive constant, the supremum is taken over all models from \(\A\), and infimum - over all possible estimates of the parameter \(c\).

We summarize the steps discussed above in the following  algorithm.\vsp

\begin{bclogo}[couleur=blue!15, logo=\bccrayon]
{Algorithm 1: Estimation of \(a\) and \(c\)}
\begin{algorithm}[H]
\SetAlgoLined
\KwData{\(n\) observations \(A_{\infty,1}, ... ,A_{\infty, n}\) of the integral \(A_{\infty}=\int_{\R_{+}} \exp\{-\xi_{s}\} ds\),\\ where \(\xi=\left( \xi_{t} \right)_{t \geq 0}\) is a L{\'e}vy process with unknown L{\'e}vy triplet \(\left( c, 0, \nu \right)\).}
\vspace{0.3cm}
Take  \(V_{n}= \kappa \log(n)\) with \(\kappa<1/(2\gamma)\), fix \(\eps \in (0,1)\) and \(u > -1\). 

Take the values \(\alpha_{1}, ..., \alpha_{M}\) on the equidistant grid on the set \([\eps, 1]\) with a step \(\Delta\). 

Define a  function \(w (\cdot) \geq 0\) supported on \([\eps,1]\).  Denote \(v_{m,n}:= \alpha_{m} V_{n}\).

\begin{enumerate}
\item
Estimate \(A_{\infty}^{s}\)  for \(s=u_{j}+\ii v_{m,n}, \; m=1..M, \) where \(u_{1}=u\) and \(u_{2}=u-1\) 
\begin{eqnarray*}
\label{step1}
	\widehat	\E_{n} \left[
		A_{\infty}^{u_{j}+\ii v_{m,n}}
	\right]
	=
	\frac{1}{n}
	\sum_{k=1}^{n} A_{\infty, k}^{u_{j}+\ii v_{m,n}}, \qquad m=1..M, \quad j=1,2.
\end{eqnarray*}
\item Estimate \(\psi(u+\ii v_{m,n})\) by 
\begin{eqnarray*}
\label{moment2}
	\hat{\psi}_{n}(u+\ii v_{m,n})
	=
	\left(
		u+\ii v_{m,n}
	\right)
\;
\frac{
	\widehat\E_{n} \left[
		A_{\infty}^{\left(u-1\right)+\ii v_{m,n}}
	\right]
}{
	\widehat\E_{n} \left[
		A_{\infty}^{u+\ii v_{m,n}}
	\right]
},  \qquad m=1..M.
\end{eqnarray*}
\item[3.] Estimate \(c\) by the solution of the optimization problem \eqref{hatcopt}, 

which is explicitly given by 
\begin{eqnarray*}
\label{opt}
	\hat{c}_{n} :=  
		 \frac{
	\sum_{m=1}^{M} 
		  w(\alpha_{m})  \alpha_{m} \: \Imm \hat\psi_{n} (u+\ii v_{m,n}) 
}
{	
	V_{n} \: \cdot \: \sum_{m=1}^{M} 
		  w(\alpha_{m})  \alpha_{m}^{2} 
}.
\end{eqnarray*}
\item[4.] Estimate \(a\) by the solution of the optimization problem \eqref{hataopt}, 

which is explicitly given by 
\begin{eqnarray*}
\label{hata}
	\hat{a}_{n} := 
		\frac{
		\sum_{m=1}^{M} w (\alpha_{m}) 	\Ree \hat\psi_{n}(u+\ii v_{m,n}) 
	}
	{
		\sum_{m=1}^{M}
		w (\alpha_{m})
	}
		-\hat{c}_{n} u.
\end{eqnarray*}
\end{enumerate}
\end{algorithm} 
\end{bclogo}

\subsection{Recovering the L{\'e}vy measure $\nu$}
As the result of the algorithm described below, we obtain the estimates \(\hat{c}_{n}\) and \(\hat{a}_{n}\) of the parameters \(c\) and \(a\).  In this subsection, we present the  algorithm for estimation the L{\'e}vy measure \(\nu\). 

First, we take points \(s=u+\ii \alpha_{m} V_{n}\), where \(\alpha_m,\) \:\(m=1..M\), belong to the interval \([-1, 1]\). The construction of the estimates \(\hat\psi_{n}(s)\) remains the same as in Section~\ref{Laplace}.  Next,  looking at \eqref{estproc}, we define an estimate \(\F_{\bar\nu}(-v)\) for \(v=\alpha_{m} V_{n}, \: m=1..M\) by
\begin{eqnarray}
\label{step5}
	 \hat\F_{\bar\nu}(-v)  =  - \hat\psi_{n}(u+\ii v) + \hat{c}_{n} (u+\ii v)   +  \hat{a}_{n}.
\end{eqnarray}
The last step is to recover the measure \(\nu\) from the estimator of the Fourier transform of the measure \(\bar\nu\). Motivated by the inverse Fourier transform formula, we propose the following nonparametric estimator of the measure \(\nu\):
\begin{eqnarray}
\label{step6}
	\tilde{\nu}(x)  &=&  \frac{1}{2 \pi} e^{u x} \int_{\R} e^{\ii v x } \hat\F_{\bar\nu}(-v) \K (v h_{n})  dv,
\end{eqnarray}
where \(\K\) is a regularizing kernel supported on \([-1,1]\) and \(h_{n}\) is a sequence of bandwidths which tends to \(0\) as \(n \to \infty\). The formal description of the algorithm is given below.


\begin{bclogo}[couleur=blue!15, logo=\bccrayon]
{Algorithm 2: Estimation of \(\nu\)}
\begin{algorithm}[H]
\SetAlgoLined
\KwData{\(n\) observations \(A_{\infty,1}, ... ,A_{\infty, n}\) of the integral \(A_{\infty}=\int_{\R_{+}} \exp\{-\xi_{s}\} ds\),\\ where \(\xi=\left( \xi_{t} \right)_{t \geq 0}\) is a L{\'e}vy process with unknown L{\'e}vy triplet \(\left( c, 0, \nu \right)\).
The estimates \(\hat{a}_{n}\) and \(\hat{c}_{n}\) are described in Algorithm 1.} 

\vspace{0.3cm}
Take the values \(\alpha_{1}, ..., \alpha_{M}\) on the equidistant grid on the set \([-1, 1]\) with a step \(\Delta\). 

Denote \(v_{m,n}:= \alpha_{m} V_{n}\).

\vspace{0.3cm}
Define  a regularizing kernel \(\K\)  supported on \([-1,1]\), and a (large enough) number \(h\).

\vspace{0.3cm}
\begin{enumerate}
\item[1-2] The first two steps coincide with given in Algorithm 1.
\item[3.]  Estimate  \(\F_{\bar\nu}(-v_{m,n}) \) for \(\bar{\nu} (dx) =e^{-ux} \nu (dx)\)  by
\begin{eqnarray*}
	 \hat\F_{\bar\nu}(-v_{m,n})  =  - \hat\psi_{n}(u+\ii v_{m,n}) + \hat{c}_{n} \cdot (u+\ii v_{m,n})  +  \hat{a}_{n}, \, m=1..M.
\end{eqnarray*}
\item[4.]
Estimate \(\nu\) by
\begin{eqnarray*}
	\hat{\nu}(x)  &=&  e^{u x}  \frac{\Delta}{2 \pi }  \sum_{m=1}^{M} e^{\ii v_{m,n} x } \hat\F_{\bar\nu}(-v_{m,n}) \K (v_{m,n} h).
\end{eqnarray*}

\end{enumerate}
\end{algorithm}
\end{bclogo}
Some theoretical and practical aspects of this algorithm are discussed in Sections \ref{secsim} and \ref{theory}.

\begin{rem}
It is a worth mentioning that the estimation algorithms 1 and 2 can be applied to more general situation when the process \(\T_{t}\) is a difference between two subordinators, i.e., \(\T_{t} = \T^{+}_{t} + \T^{-}_{t}\), where  \(\T^{+}\) and \(\T^{-}\) are the processes of finite variation with L{\'e}vy measures \(\nu^{+}\) and \(\nu^{-}\) concentrated on \(\R_{+}\) and \(\R_{-}\) resp. In fact, in this case, the formula \eqref{estproc} still holds with \[\nu(dx) = \I \{x>0\} \nu^{+} (dx) + \I \{x<0\} \nu^{-}(dx).\] 
Therefore, the consequent estimation of \(c\), \(a\) and the Fourier transform of the measure \(e^{-u x} \nu(dx)\), as well as the estimation of \(\nu\) are still possible. 

Theoretical results under  the assumptions (A2) and (A3)  remain the same. The example from Section~\ref{further} can be naturally extended to 
\[
\nu(x) = I_{x>0} \sum_{j=1}^{M} \sum_{k=1}^{{m_{j}}} \alpha_{jk} x^{k-1} e^{-\rho_{j}x} + I_{x<0} \sum_{j=1}^{\tilde{M}} \sum_{k=1}^{{\tilde{m}_{j}}} \tilde{\alpha}_{jk} x^{k-1} e^{-\tilde{\rho}_{j}x}
\]
with \(M, \tilde{M}, m_{j}, \tilde{m}_{j} \in \N\), \(\rho_{j}, \tilde\rho_{j}>0\), \(\alpha_{jk}, \tilde\alpha_{}{jk}>0\). Note that Assumption (A2) is already checked in Theorem~1 from \cite{Kuznets2}.
\end{rem}

\section{Simulation study} 
\label{secsim}

\sec{Example 1.} 
Consider the subordinator \(\T_{t}\) with the L{\'e}vy density 
\begin{eqnarray}
\label{nu}
\nu(x)= ab \exp\{-bx\}\: I\{x>0\}, \quad a, b >0.
\end{eqnarray}
For this subordinator, the integral \(A_{\infty}\) is finite for any \(\sigma\), see \cite{CPY}.
The Laplace exponent of \(\xi\) is given by
\begin{eqnarray}
\label{phiss2}
\psi(z) = z \left(
		c - \frac{1}{2} \sigma^{2} z + \frac{a}{b+z}
	\right).
\end{eqnarray}
As for the distribution properties of \(A_{\infty}\),  the density function of \(A_{\infty}\) satisfies the following differential equation 
\begin{multline}
 -\frac{\sigma^{2}}{2} x^{2} k''(x) + 
 \left[
 	\left(
		\frac{\sigma^{2}}{2} (3-b) +c
	\right)
	x
	-1
\right] k'(x) \\
+
\left[
	\left( 
		1-b
	\right)
	\left(
	 	\frac{\sigma^{2}}{2} +c
	\right)
	-a+\frac{b}{x}
 \right] k(x) = 0,
 \label{eq}
\end{multline}
see \cite{CPY}.  Some typical situations are given below:
\begin{enumerate}
\item  In the case \(c=0, \: \sigma=0\) (pure jump  process), this equation has a solution
\begin{eqnarray}
\label{kx1}
	k_{1}(x) = C x^{b} e^{-a x} \: I\{x>0\},
\end{eqnarray}
and therefore \(A_{\infty} \eqd  G(b+1, a)\), where \(G(\alpha, \beta)\) is  a Gamma distribution with shape parameter \(\alpha\)  and rate \(\beta\). 
\item If \(c>0, \: \sigma=0\)  (pure jump process with drift), then 
\begin{eqnarray}
\label{kx2}
	k_{2}(x) = C x^{b} (1-c x)^{(a/c) -1 } \: I\{0< x<1/c\}.
\end{eqnarray}
In this situation \(A_{\infty} \eqd  B(b+1, a/c) / c\), where \(B(\alpha, \beta)\) is a Beta - distribution.
\item In the case \(c \ne 0, \: \sigma \ne 0\),  the equation \eqref{eq} also allows for the closed form solutions. Assuming for simplicity \(\sigma^{2}/2 =1\), \(c=-(b+1)\), we get the solution of \eqref{eq} in the following form:
\begin{eqnarray}
\label{kx}
	k_{3}(x)=C\: x^{b-1/2} \exp\left\{\frac{1}{2x}\right\} I_{\mu}\left(\frac{1}{2x}\right),
\end{eqnarray}
where we denote by \(I_{\mu}\) the modified Bessel function of the first kind, \(\mu=\sqrt{a+1/4}\), and the constant \(c\) is later chosen to guarantee the condition \(\int_{0}^{\infty}k_{3}(x) dx =1\). \vsp
\end{enumerate}

For the numerical study, we assume that the data follows the model \eqref{Ainfty} where the process \(\xi_{t}\) is defined by \eqref{xi} with 
\(c=1.8\), \(\sigma=0\),  and the subordinator \(\T_{t}\) has a L{\'e}vy density in the form \eqref{nu} with \(a=0.7\), \(b=0.2\). The values of the integral \(A_{\infty}\) are simulated from the Beta-distribution, see \eqref{kx2}. 

On the first step, we estimate  \(A_{\infty}^{s}\)  for \(s=u+\ii v\) with \(u=29\) and \(u=30\) and  \(v\) from the equidistant grid between \(-30\) and \(30\). Next, we estimate the Laplace exponent by the formula \eqref{step2}. One can visually compare the proposed estimator and the theoretical value \((c+a/(b+s))*s\) looking at  Figure \ref{plot1}. 
\begin{figure}
\begin{center}
\includegraphics[width=0.6\linewidth ]{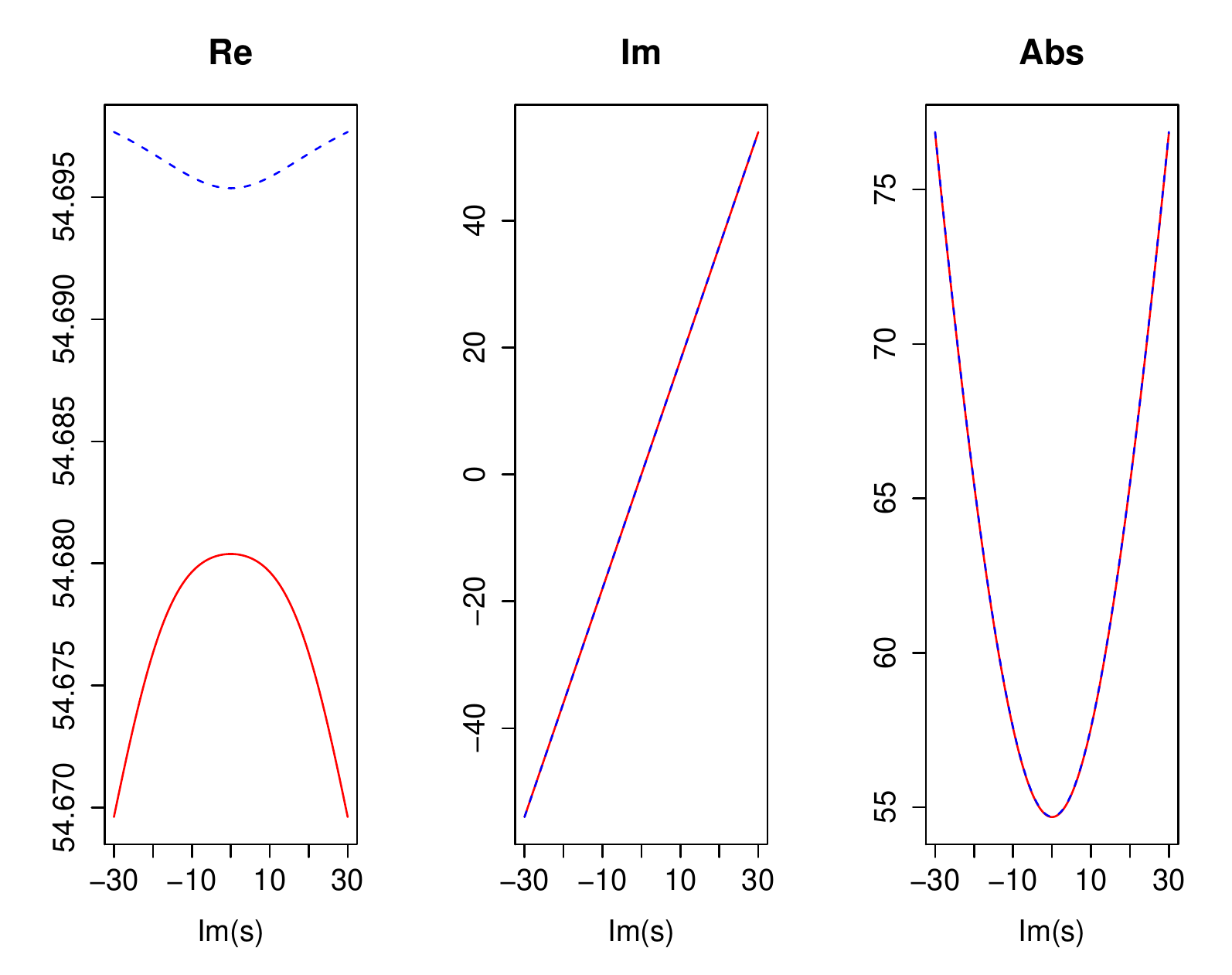}\caption{Plots of theoretical (blue dashed) and empirical (red solid) Laplace exponents for Example 1. Graphs present real, imaginary parts and absolute values. In spite of visual distinction in the real parts, the difference between theoretical and empirical Laplace exponents is quite small.\label{plot1}}
\end{center}
\end{figure}

\begin{figure}
\begin{center}
\includegraphics[width=0.6\linewidth ]{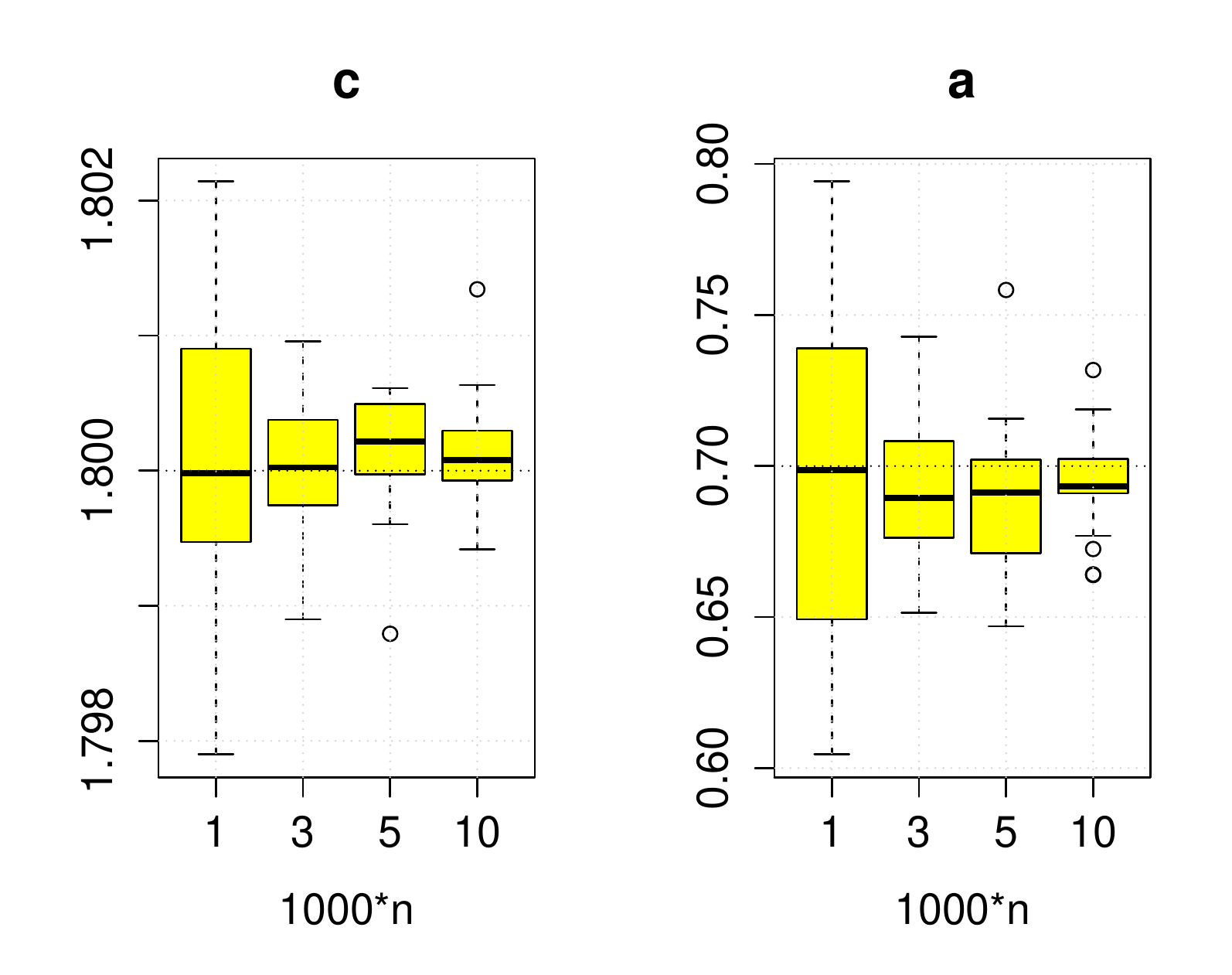}\caption{Boxplots for the estimates of \(c\) and \(a\) for different values of \(n\) based on 25 simulation runs.\label{plot2}}
\end{center}
\end{figure}
Estimation of the parameters \(c\) and \(a\) is provided by  \eqref{hatc} and \eqref{hata} resp. The boxplots of this estimates are presented on Figure \ref{plot2}. \vsp

\sec{Example 2.}  Consider the compound Poisson process 
\[
	\xi_{t} = -\log q \left(\sum_{k=1}^{N_{t}} \eta_{k} \right),
\]
where \(q \in (0,1)\) is fixed, \(N_{t}\) is a Poisson process with intensity \(\lambda\) and \(\eta_{k}\) are i.i.d. random variables with a distribution \(\L\).
It is a worth mentioning that the integral \(A_{\infty}\) allows the representation 
\begin{eqnarray*}
	A_{\infty} = \int_{0}^{\infty} q^{-\xi_{t}} dt = \sum_{n=0}^{\infty} q^{-n} \left(T_{n+1} - T_{n}\right),
\end{eqnarray*}
where \(T_{n}\) is the jump time \(T_{n} = \inf\left\{ t: N_{t} =n \right\}\). Note that if \(\eta_{k}\) takes only positive values then \(-\xi_{t}\) is a subordinator. For the overview of the properties of the integral \(A_{\infty}\) in the particular case \(\L \equiv 1\) (that is, \(\xi_{t}\) is a Poisson process up to a constant), we refer to \cite{BertoinYor}.

Fix some positive \(\alpha\) and consider the case when \(\L\) is  the standard Normal distribution truncated on the interval \((\alpha, +\infty)\).  The density function of \(\L\) is given by \[p_{\L} (x)= p(x) / (1-F(\alpha)),\] where \(p(\cdot)\) and \(F(\cdot)\) are pdf and cdf of the standard Normal distribution. In this case, the Laplace exponent  of \(\xi_{t}\) is equal to 
\begin{eqnarray*}
	\psi(s) = \lambda \left[
		1 -  \frac{	
			1-F\left(\alpha + (\log q) s\right)
		}
		{
			1 - F\left(\alpha\right)
		}\;
		\exp \left\{
		-
			\frac{
				\left(
					\log q 
				\right)^{2}
				s^{2}
			}{2}
		\right\}
	\right],
\end{eqnarray*}
where the function \(F(\cdot)\) in the complex point \(z\) can be calculated from the error function: 
\[
F(z) := \frac{1}{2} \left( \erf\left( \frac{z}{\sqrt{2}} \right) + 1 \right), \quad \mbox{where} \quad 
\erf(z) = \frac{2}{\sqrt{\pi}}\int_{0}^{z} e^{-s^{2}} ds.
\]
In this example, we aim to  estimate the L{\'e}vy measure of the process \(\-\xi_{t}\), which is equal to 
\[
\nu(dx)  = \frac{\lambda}{1-F(\alpha)} \: p(x)  I\{x>\alpha\} dx.
\]
For the numerical study, we take \(q=0.5, \alpha=0.1\), and \(\lambda=1\). First, we estimate the Laplace exponent by  \eqref{step2}. The quality of estimation  at the complex points \(s=u+\ii v\) with \(u=1\) and \(v \in [-5,5]\) can be visually checked on Figure~\ref{fig3}. 
\begin{figure}
\begin{center}
\includegraphics[width=0.6\linewidth ]{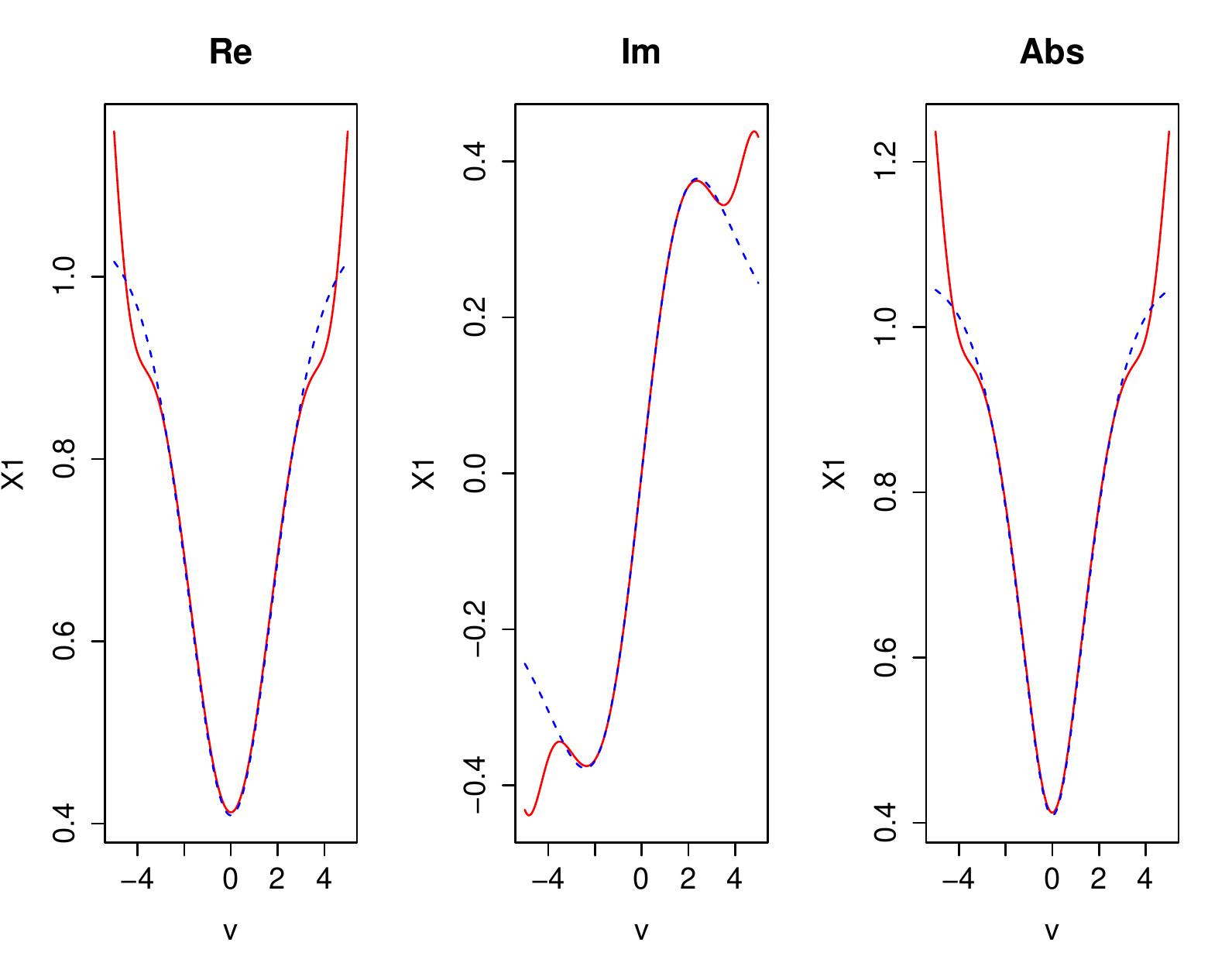}\caption{Plots of theoretical (blue dashed) and empirical (red solid) Laplace exponents for Example 2. Graphs present real, imaginary and absolute values. For  \(v \in[-3,3]\) the curves are visually indistinguishable. \label{fig3}}
\end{center}
\end{figure}
\begin{figure}
\begin{center}
\includegraphics[width=0.6\linewidth ]{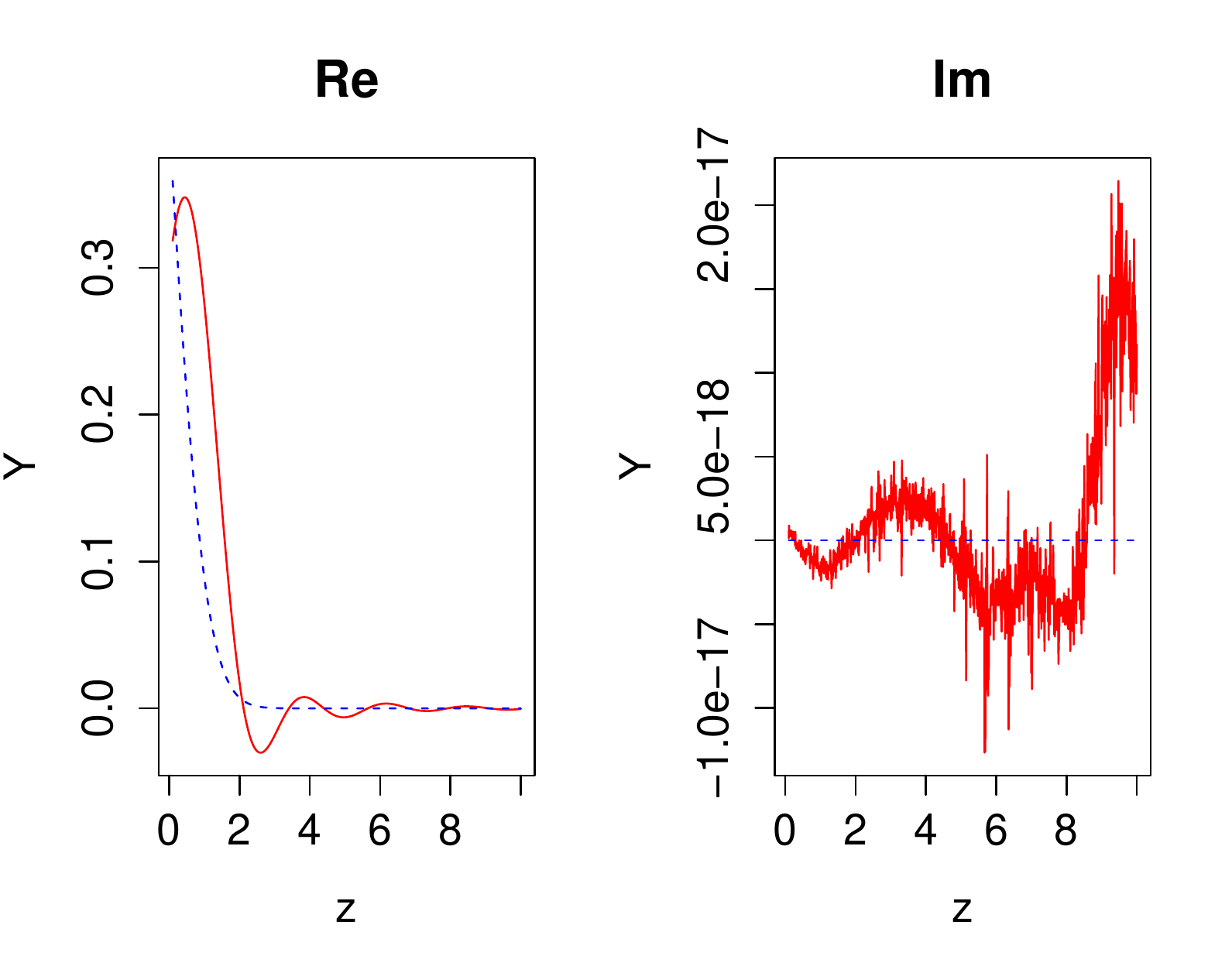}\caption{Plots of the L{\'e}vy measure  (blue dashed line) and its estimate (red solid) depicted for real (left) and imaginary (right) parts. Note that the values on the right plot are quite small. \label{fig4}}
\end{center}
\end{figure}

Next, we proceed with the estimation of the Fourier transform of the measure \(\bar{\nu}(x):=e^{-u x} \nu(x)\)
of the L{\'e}vy measure by applying \eqref{step5}. For the last step of the Algorithm~2, reconstruction of the L{\'e}vy measure by \eqref{step6},  we follow  \cite{Belomest2011}  and take the so-called flat-top kernel, which is defined as follows:	
\begin{eqnarray*}
\K(x)=
\begin{cases}
1, & |x|\leq 0.05, \\
\exp\left( -\frac{e^{-1/(|x|-0.05)}}{1-|x|} \right), & 0.05<|x|<1,\\
0, &  |x|\geq 1.
\end{cases}
\end{eqnarray*}
The quality of the resulted  estimation  is given on Figure~\ref{fig4}.

\section{Theoretical study}
\label{theory}
\begin{thm}
\label{thmphi}
Consider the model \eqref{Ainfty} with L{\'e}vy process \(\xi\) in the form \eqref{xi} satisfying the assumptions (A1) - (A3). 
Let the sequence \(V_{n}\) tend to \(\infty\) and moreover satisfy the assumption 
\begin{eqnarray}
\label{lambda}
     \Lambda_{n} := V_{n} \exp\left\{\gamma V_{n}\right\} \sqrt{\log V_{n}} = o \left(\sqrt{\frac{n}{\log(n)}} \right), \; n \to \infty,
\end{eqnarray}
where the constant \(\gamma\) is introduced in (A2).  Then there exists a set \(\W_{n}\) such that  \(\PP\{\W_{n}\}>1 - \alpha n^{-1-\delta}\) (with some positive \(\alpha\) and \(\delta\)) and 
\begin{eqnarray}
\label{Wn}
\W_{n}
 \subset
\Biggl\{
\sup_{v \in [\eps V_{n}, V_{n}]}
\left|
	\hat\psi_{n}(u+\ii v) - \psi(u + \ii v)
\right|
\leq \beta \: \Lambda_{n}\sqrt{\frac{\log(n)}{n}} 
\Biggr\} 
\end{eqnarray}
where \(\beta>0\) and \(u=u^{\circ}\) was introduced in (A2).
\end{thm}
\begin{rem}
The condition \eqref{lambda} fulfills for instance for \(V_{n} = \kappa \log (n)\) with \(\kappa<1/\left(2 \gamma \right)\).
\end{rem}
\begin{proof}
\ttt{1.} Denote 
\begin{eqnarray*}
	J(s) : = \Bigl(	\widehat	\E_{n} \left[
		A_{\infty}^{s}
	\right]
	- 
	\E \left[
		A_{\infty}^{s}
	\right]
\Bigr) / \E \left[
		A_{\infty}^{s}
	\right],		
\end{eqnarray*}
where \(s=u+\ii v\). In this notation, 
\begin{multline}
\label{proof1}
\Biggl| 
 	\psi(s) - \hat\psi_{n}(s) 
\Biggr| = 
\left|
s \Biggl(
\frac{
	\E \left[
		A_{\infty}^{s-1}
	\right]
}{
	\E \left[
		A_{\infty}^{s}
	\right]
}
-
\frac{
	\widehat\E_{n} \left[
		A_{\infty}^{s-1}
	\right]
}{
	\widehat\E_{n} \left[
		A_{\infty}^{s}
	\right]
}
\Biggr)
\right|
\\=
\Biggl|
	s \frac{
	\E \left[
		A_{\infty}^{s-1}
	\right]
}{
	\E \left[
		A_{\infty}^{s}
	\right]
}
\Biggr|
\; \cdot \;
\Biggl|
	\frac{J(s) - J(s-1)}{1+J(s)} 
\Biggr|.
\end{multline}
By \eqref{moment}, the first term is equal to \(|\psi(s)|\), and therefore by \eqref{upp} it is bounded by \(C_{1} \Imm (s)\) for \(\Imm(s)\) large enough with some \(C_{1}>0\).  As for the second term, we firstly note that 
\begin{eqnarray*}
	\Biggl|
	\frac{J(s) - J(s-1)}{1+J(s)} 
\Biggr|
\leq 
\frac{|J(s)|+|J(s-1) |}{1-|J(s)|}.
\end{eqnarray*}
The aim of the further proof is to show that  the right hand side in the last inequality is bounded by \(\sqrt{\log(n)/n}\) on  a probability set with desired properties.

\ttt{2.}  Proposition~\ref{ExpBounds} yields that there exists such set \(\W_{n}\) of probability mass larger than \(1- \alpha n^{-1-\delta}\), such that it holds on this set 
\begin{eqnarray}
\label{supim}
	\sup_{s: \: \Imm(s) \in I_{n}} \Bigl| \widehat 
	\E_{n} \left[
		A_{\infty}^{s}
	\right]
	-
	\E 
	\left[
		A_{\infty}^{s}
	\right]
	\Bigr|
	\lesssim \sqrt{\log (V_{n})  \log(n) / n}, \qquad n \to \infty,
\end{eqnarray}
where \(\alpha\) and \(\delta\) are positive, and  \(I_{n}:=[\eps V_{n}, V_{n}]\). In fact, direct application of   Proposition~\ref{ExpBounds} with  a weighting function \(w^{*}(x) :=\log^{-1/2} \left(e+|x|\right)\) gives
\begin{eqnarray*}
\sup_{v \in I_{n}} \Bigl| 
	\widehat \E_{n} \left[
		A_{\infty}^{u+\ii v}
	\right]
	-
	\E 
	\left[
		A_{\infty}^{u+\ii v}
	\right]
	\Bigr|
	&\leq&
	 \sup_{v \in I_{n}} \left[ 
	 \Bigl|
	 	\frac{w^{*}(v)}{\inf_{x \in I_{n}} w^{*}(x)} 
	\Bigr|
	\right.
	\\
	 && 	\left. \hspace{2cm}
	 \cdot
	 \Bigl|
	 \left(
		\widehat{\E}_{n} \left[
			A_{\infty}^{u+\ii v}
		\right]
		-
		\E 
		\left[
			A_{\infty}^{u+\ii v}
		\right]
	\right)
	\Bigr| \right]\\
	&\leq &
	\sqrt{\log\left(e+V_{n}\right)}\\ 
	&&\hspace{1cm} \cdot
	\sup_{v \in I_{n}} \Bigl|
	 	w^{*}(v)
	 \left(
		\widehat {\E}_{n} \left[
			A_{\infty}^{u+\ii v}
		\right]
		-
		\E 
		\left[
			A_{\infty}^{u+\ii v}
		\right]
	\right)
	\Bigr|\\
	 &\lesssim& \sqrt{\log (V_{n})  \log(n) / n}, \qquad n \to \infty.
\end{eqnarray*}

\ttt{3.} Formula \eqref{supim} in particularly means the following inequlity holds on the set \(\W_{n}\) 
\begin{eqnarray}
\label{supJ}
	\sup_{s: \: \Imm(s) \in I_{n}}  |J (s)| \lesssim  \exp\{ \gamma V_{n}\}\; \sqrt{\log (V_{n})  \log(n) / n} \qquad n \to \infty, .
\end{eqnarray}
It is worth mentioning that under the assumption \eqref{lambda}, 
\begin{eqnarray}
\label{J}
	\sup_{s: \: \Imm(s) \in I_{n}}  |J (s)|  \to 0 \quad \mbox{ as } \quad n \to \infty.
\end{eqnarray}
Substituting   \eqref{supJ} into \eqref{proof1} and taking into account \eqref{upp} and \eqref{J}, we arrive at the following bound for the quality of the estimate \(\hat\psi_{n}(s)\):
\begin{eqnarray*}
\Bigl| 
 	\psi(s) - \hat\psi_{n}(s) 
\Bigr| \lesssim 
V_{n} \exp\{ \gamma V_{n}\}\; \sqrt{\log (V_{n})  \log(n) / n}
, 
\end{eqnarray*}
which holds on the set \(\W_{n}\). This observation completes the proof.
\end{proof}

\begin{thm}
\label{thm3}
Consider the setup of Theorem~\ref{thmphi} and take \(V_{n} = \kappa \log(n)\) with \(\kappa<1/(2\gamma)\). 
Then  it holds
\begin{eqnarray}
\label{emb}
\W_{n}
 \subset
\Biggl\{
\left| 
	\tilde{c}_{n} - c 
\right| 
\leq  \frac{\zeta_{1}}{\log^{r+2}(n)}
\Biggr\}
\quad \mbox{and} \quad 
\W_{n}
\subset
\Biggl\{
\left| 
	\tilde{a}_{n} - a
\right| 
\leq \frac{\zeta_{2}}{\log^{r+1}(n)}
.
\Biggr\}
,
\end{eqnarray}
where \(s\) is introduced in (A3), the set \(\W_{n}\) is defined in Theorem~\ref{thmphi}, and \(\zeta_{1}, \zeta_{2}>0\).
\end{thm}

\begin{proof}
\ttt{1.}
First note that the estimate \eqref{optim} can be rewriten as 
\begin{multline*}
	\tilde{c}_{n} = \int_{0}^{\infty} w_{n}^{*}(v) \Imm \hat{\psi}_{n} (u+\ii v) dv,\\
\mbox{where} \quad 
	 w_{n}^{*}(v) = \frac{ w_{n}(v) v } {\int w_{n}(y) y^{2} dy} = \frac{1}{V_{n}^{2}} w^{*}\left( \frac{v}{V_{n}}\right),
\end{multline*}
where \( w^{*} (x) = \left(w(x) x \right) / \left( \int w(y)  y^{2} dy \right)\).  Next, consider the following ``theoretical counterpart'' of the estimate \(\tilde{c}_{n}\):
\begin{eqnarray*}
	\bar{c}_{n} :=  
\int_{0}^{\infty} w_{n}^{*}(v) \Imm \psi (u+\ii v) dv,
\end{eqnarray*}
and note that 
\begin{eqnarray}
\label{tildec}
  \left| \hat{c}_{n} - c \right| \leq \left| \hat{c}_{n} - \bar{c}_{n} \right|  + \left| \bar{c}_{n} - c \right|.
\end{eqnarray}
The first summand in the right hand side of \eqref{tildec} is bounded on the set \(\W_{n}\) for \(n\) large enough:
\begin{multline}
\label{c1}
	 \left| \hat{c}_n -\bar{c}_n \right| \leq \ A \: \Lambda_{n}\sqrt{\frac{\log(n)}{n}} \frac{1}{V_{n}},
\\
\mbox{where} \quad
	 A:= V_{n}\beta \;
	 \left| 
	 \frac{ \int w_{n}^{*}(v) v dv} {\int w_{n}^{*}(v) v^{2} dv}
	\right|
	= 
	 \left| 
	 \frac{ \int_{\eps}^{1} w^{*}(v) v dv} {\int_{\eps}^{1} w^{*}(v) v^{2} dv}
	\right|
\end{multline}
doesn't depend on \(n\).
As for the second term, using \(\int w_{n}^{*} (v) v  dv =1\), we get 
\begin{eqnarray*}
 \left| \bar{c}_n - c \right| =
 \left|
	 \int_{0}^{\infty} w_{n}^{*}(v) \Bigl[
 		\Imm \hat{\psi}_{n} (u+\ii v) dv - c v 
	\Bigr] dv 
\right|
= 
 \left|
	 \int_{0}^{\infty} w_{n}^{*}(v) 
 		\Imm \F_{\bar{\nu}} (-v) dv
\right|.
\end{eqnarray*} 

 Applying Lemma~\ref{lemlem} with \(w_{n}^{*}(v) = V_{n}^{-2} w_{1}^{*}\left( v/V_{n}\right)\),
we get 
\begin{eqnarray}
\label{c2}
\Bigl|
	 \int_{0}^{\infty} w_{n}^{*}(v) 
 		\F_{\bar{\nu}} (v) dv
\Bigr| \lesssim V_{n}^{-(r+2)}, \qquad n \to \infty
.\end{eqnarray}

 Substituting \eqref{c1} and \eqref{c2} into \eqref{tildec}, and bearing in mind our choice of \(V_{n}\), we complete the proof of the first embedding in \eqref{emb}.
  
\ttt{2.} 
Without limitations we can assume that \(\int_{\R+} w_{n}(v) dv = \int_{\eps}^{1} w(v) dv =1\).   The second embedding directly follows from Theorem~\ref{thmphi} and the first part of this proof, because 
\begin{eqnarray*}
	\left| \tilde{a}_{n} - a \right| 
	&=& 
	\left| 
	\left[
		\int_{\R_{+}}
			 w_{n}(v)  \Ree \hat\psi_{n}(u+\ii v) dv 
		-\tilde{c}_{n} u
	\right]
	\right.
	\\&&\hspace{1.5cm} \left. - 
	\left[
		\int_{\R_{+}}
			 w_{n}(v) \Bigl(  \Ree \psi(u+\ii v) +\Ree \F_{\bar\nu} (-v) \Bigr) dv 
		- c u
	\right]
	\right|\\
	&\leq& 
	\left| \tilde{c}_n -c \right| u
	+ 
	\left|
		\int_{\R_{+}}
			 w_{n}(v) \Bigl(  \Ree \hat\psi_{n}(u+\ii v) -  \Ree \psi(u+\ii v)  \Bigr) dv 
	\right|
	\\
	&& \hspace{4cm}
	+ 
	\left|
		\int_{\R_{+}}
			 w_{n}(v) \Ree \F_{\bar\nu} (-v)  dv 
	\right|
\\
	&\lesssim& 
	\frac{\zeta_{1}}{\log^{r+2}(n)} +\beta \: \Lambda_{n}\sqrt{\frac{\log(n)}{n}}  + \frac{\lambda}{\log^{r+1} (n)}\\
	&\lesssim& \frac{1}{\log^{r+1} (n)}, \qquad n \to \infty,
\end{eqnarray*}
where \(\lambda>0\).  Note that here we use the inequality  \[\left|
		\int_{\R_{+}}
			 w_{n}(v) \Ree \F_{\bar\nu} (-v)  dv 
	\right|
	\lesssim \log^{-(r+1)} (n),
\]
which follows by applying Lemma~\ref{lemlem} to \(w_{n}(v) = V_{n}^{-1} w_{n}(v/V_{n})\). This completes the proof. 
\end{proof}

\begin{thm}
\label{lower}
Let \(\A\) be a set of functions that satisfy assumptions (A1) - (A3). Then it holds
\[
\varliminf_{n \to \infty} 
\inf_{\tilde{c}_{n}^{*}} \sup_{\A}  
	\PP\Biggl\{
\left|
	\tilde{c}_{n}^{*} - c 
\right|
\geq \zeta_{3} \log^{-(r+2)}(n)
\Biggr\}  > 0,
\]
where \(\zeta_{3}\) is some positive constant, the supremum is taken over all models from \(\A\), and infimum - over all possible estimates of the parameter \(c\).
\end{thm}

\begin{proof}
We follow the general reduction scheme, which can be found in \cite{KorTsyb} and \cite{Tsyb}. Consider a class of L{\'e}vy processes \(\A\) that satisfies the assumptions (A1)-(A3). There exist two L{\'e}vy process \(\xi_{0}\) and \(\xi_{1}\) from \(\A\), having L{\'e}vy triplets \(\left(c_{0}, 0, \nu_{0}\right) \), \(\left(c_{1}, 0, \nu_{1} \right)\), Laplace exponents \(\phi_{0}\), \(\phi_{1}\), exponential functionals with densities \(p_{0}\), \(p_{1}\) and Mellin transforms \(M_{0}\), \(M_{1}\),  such that it holds simultaneously
\begin{enumerate}
\item 
the L{\'e}vy triplets are related by the following identities: 
\begin{eqnarray}
\label{nu0nu1}
	c_{0} - c_{1} = 2\delta, \qquad \nu_{0} (x)  -	\nu_{1} (x)  = 2 \delta  K'_{h} (x), 
\end{eqnarray}
where \(\delta>0\), \(K_{h} (x) = h^{-1} K \left( h^{-1} x \right)\) for any \(x \in \R\) and some  \(h>0\), and \(K \in L^{1}(\C)\) satisfy \(\F_{K}(z)=-1\) for \(z\) with \(\Ree(z) \in [-1,1]\) and polynomical decay  \(\left| F_{K}(z) \right| \lesssim \left| Re(z) \right|^{-\eta}\) as \(|z| \to \infty\).
\item  the density of one of the functionals, say  the first one,  decays at most polynomially, i.e., there exists \(m \in \N\) such that 
\[
	p_{0}(x) \gtrsim (1+x)^{-2m}, \qquad x \to +\infty.
\]
\item  \(M_{0}(s) \) and \(M_{1}(s) \) coincide on the lines \(s=u^{(k)}+\ii v, \: k=1,2,\) for  \(u^{(1)}=3/2\), \(u^{(2)}=m+3/2\),  and any \(v\). Moreover, the asymptotics of the Mellin transforms along these lines is given by (A2), i.e.,
\[ 
	|M_{j} (u^{(k)} + \ii v) |\asymp \exp\{-\gamma^{(k)} |v|\}, \qquad  \mbox{as}\quad \; v \to \infty,
\]
with some \(\gamma^{(k)}>0, \; k=1,2, \: j=0,1.\)

\end{enumerate}
Let the exponential functionals of these L{\'e}vy processes have distribution laws \(\PP_{0}\) and \(P_{1}\).
\[   
	\chi^{2}_{n}(1|0) := \chi^{2} ( \PP_{1}^{\otimes n} | \PP_{0} ^{\otimes n})
	 \leq \exp\left\{n \chi^{2} ( \PP_{1} | \PP_{0} )\right\} -1, 
\]
see Lemma~5.5 from \cite{BelReiss}. The aim is to show that there exist a constant \(C>0\) such that \(\chi^{2}_{n}(1|0) <C\); after that the desired result will immediately follow, see Part~2 and especially Theorem~2.2 from \cite{Tsyb}.

Our choice of the models leads to the following estimate of the chi-squared distance between \(\PP_{1}\) and \(\PP_{2}\): 
\begin{multline}
\chi^{2}( 1 | 0 ) :=\chi^{2} ( \PP_{1} | \PP_{0} )= \int_{\R_{+}} \frac{\left(  p_{1}(x) - p_{0} (x) \right)^{2}}{ p_{0}(x)} \\
\lesssim \int_{\R_{+}} (1 + x^{2m})  \left(  p_{1}(x) - p_{0} (x) \right)^{2} dx.
\end{multline}
By Lemma~\ref{lemmelin} we get that 
\begin{multline}
	\chi^{2}( 1 | 0 ) 
	\lesssim 
	\Delta(0) + \Delta(m),
	\\ \mbox{where} \quad 
	\Delta (\cdot):=
	\int_{-\infty}^{\infty} \left|
		M_{0} (\cdot+1/2+\ii v) - M_{1} (\cdot+1/2+\ii v) 
	\right|^{2} dv.
\end{multline}	 
Note that by \eqref{momenty} and our assumptions, 
\begin{eqnarray}
\label{M0M1}
	 M_{0} (s - 1) - M_{1} (s - 1) 
	  = 
	  \frac{
	  	\phi_{0}(s) - \phi_{1}(s)
	  }
	  {	
	  	s
	} M_{0}(s),
\end{eqnarray}
where  \(s=u^{(k)}+\ii v\), \(k=1,2\).  By our choice of the L{\'e}vy measures \eqref{nu0nu1} and the representation of the Laplace exponent \eqref{phis2}, we get 
\begin{eqnarray}
\nonumber
	\phi_{0} (s) - \phi_{1}(s) 
	&=&
	\left(	
	 	c_{0} - c_{1}
	\right) s
	+
	\int_{\R_{+}} \left(1-e^{-sx}\right) \nu_{0} (dx)\\
	\nonumber
	&& \hspace{3cm}
	-
	\int_{\R_{+}} \left(1-e^{-sx}\right) \nu_{1} (dx)	\\
	\nonumber
	&=& 2 \delta s +  \int _{\R_{+}} \left[
		\nu_{0} (dx) - \nu_{1}(dx)
	\right] 
	-
	\left[
		\F_{\nu_{0}} (\ii s) 
		-
		\F_{\nu_{1}} (\ii s) 		
 	\right]\\
	\nonumber
	&=& 
	2\delta s + 
	2\delta  
	\int_{\R_{+}} K'_{h} (x) dx
	- 
	2\delta	\F_{K'_{h}} (\ii s) 
\end{eqnarray}
Next, we take into account that \(\F_{K'_{h}} (y) = \ii y \F_{K_{h}}(y) = \ii y F_{K}(y h)\) for any \(y \in \C\). Therefore \(\int_{\R_{+}} K'_{h} (x) dx=\F_{K'_{h}} (0) =  0\) and moreover
\begin{eqnarray}
	\label{phi0phi1}
	\phi_{0} (s) - \phi_{1}(s)  = 2 \delta s \left( 1 + \F_{K} (\ii s h) \right).
\end{eqnarray}
Substituting \eqref{phi0phi1} into \eqref{M0M1}, we arrive at 
\begin{eqnarray*}
\label{M0M12}
	 M_{0} (s - 1) - M_{1} (s - 1)  = 2 \delta \left( 1 + \F_{K} (\ii s h) \right) M_{0}(s),
\end{eqnarray*}
and therefore 
\[
	\Delta(\cdot)
	= 
	\delta 
	\int_{\R} 
	\Bigl|
		1 
		+
		\F_{K} \Bigl(\left(-  v + \ii u^{(k)} \right) h\Bigr)
	\Bigr|^{2} * \Bigl|  M_{0}\Bigl(u^{(k)}+\ii v\Bigr)
	\Bigr|^{2} dv,
\]
where \(k=1\) if \(\cdot = 0\) and \(k=2\) if \(\cdot=m\). By our assumptions on the kernel \(K\), we get 
\[
	\Delta(\cdot) \lesssim \delta \int_{|v| > 1/h} e^{-\gamma^{(k)} |v|} dv =  \frac{\delta}{\gamma^{(k)}} e^{-\gamma^{(k)} /h}, 
\]
and therefore 
\[
	\chi^{2}( 1 | 0 ) 
	\lesssim 
	\frac{\delta}{\gamma^{*} }e^{-\gamma^{*}  /h}, \qquad \mbox{with} \quad \gamma^{*} := \min \left\{ 
		\gamma^{(1)}, \gamma^{(2)}
	\right\}.
\]	
If we choose \(\delta = h^{s+2}\) and \(h = \log^{-1}(n) \gamma^{*}  / (1+\eps)\) for any (small) \(\eps>0\), the \(\chi^{2}\) - divergence is bounded by 
\[
	\chi^{2}( 1 | 0 ) = \frac{\left(\gamma^{*}\right)^{s+2}}{(1+\eps)^{s+2}} \frac{\log^{-(s+2)}(n)}{n^{1+\eps}} \lesssim\frac{\log(C+1)}{n}
\]
for any \(C>0\) an \(n\) large enough. Therefore, 
\[
	\chi^{2}_{n} ( 1 | 0 )  \leq \exp\{n \chi^{2}(1|0)\} -1 \leq C, 
\]
and the statement of the theorem follows.

\end{proof}

\section{Appendix. Additional proofs}
\begin{lem}[Exponential inequalities for dependent sequences]
\label{EIB}
Let \( (G_k, \, k\geq 1) \) be  a sequence of centered real-valued random variables on the probability space $(\Omega,\cal F,P)$. Assume that
\begin{enumerate}
\item \(G_{k}\) is a strongly mixing sequence with the mixing coefficients satisfying
\begin{eqnarray}
\label{ALPHA_EXP_DECAY}
   \alpha_{G}(n)\leq \bar\alpha_{0}\exp\left\{- \balpha_{1} n\right\},\quad n\geq 1,\quad \bar\alpha_{0}>0,\quad \bar\alpha_{1}>0;
\end{eqnarray}
\item   $\sup_{k\geq 1}|G_k|\leq M$ a.s. for some positive \(M\);
\item the quantities 
\[
\rho _{k}:=\E\left[ G_{k}^{2}\:\left| 
2	\log G_{k} 
\right|^{2(1+\varepsilon )}\right], \quad k=1,2,\ldots,
\]%
 are finite for all \(k\) with some small \(\varepsilon>0\).
\end{enumerate} 
Then there is a positive constant \( C_{1} \) depending on \( \balpha:= \left( \balpha_{0}, \balpha_{1} \right) \) such that
$$
\PP\left\{ \sum_{k=1}^n G_k\geq \beta  \right\}\leq \exp\left[-\frac{C_{1}\beta^2 }{nv^{2}+M^{2} +M\beta \log^{2} n}\right].
$$
for all \( \beta>0 \) and \( n\geq 4, \)
where
\begin{eqnarray*}
    v^{2}\leq \sup_{k}\E[G_{k}^{2}]+C_{2}\sup_{k}\rho _{k}
\end{eqnarray*}
with \(C_{2}>0\).
\end{lem}
\begin{proof}
The proof directly follows from Theorem A.1 and Corollary A.2 from \cite{panov2013a}.
\end{proof}

The next result gives the uniform probabilistic inequality for the empirical process. This result is an analogue of Proposition A.3 from  \cite{panov2013a}, which gives the uniform inequality for the case when \(u=0\) (see below). For similar results in i.i.d. case, see \cite{NeuReiss}.

\begin{prop}
\label{ExpBounds}
Let \( Z_{j},\)  \(j=1,\ldots, n, \) be a stationary sequence of   random variables.  Define
\begin{eqnarray*}
	\p_{n}(v)&:=&\frac{1}{n}\sum_{j=1}^{n}\exp\left\{ \left( u +\ii v \right) Z_{j}\right\},
\end{eqnarray*}
where \(u \in \R_{+}\) is fixed and \(v \in \R\) varies. Let \(\p(v)\) be a characteristic function of the corresponding stationary distribution.  Let also $w$ be a positive monotone decreasing  Lipschitz function on $\mathbb{R}_{+}$ such that
\begin{equation}
\label{decreasing_w}
	0<w(z)\leq  \frac{1}{\sqrt{\log(e+|z|)}}, \quad z\in \mathbb{R}.
\end{equation}
Suppose that  the following assumptions hold:
\begin{description}
\item[(A1)] random variables \(e^{Z_{j}} \)  possess finite absolute moments of order \( p>2\).
\item[(A2)] \(Z_{j}\) is a strongly mixing sequence with the mixing coefficients satisfying
\begin{eqnarray}
\label{ALPHA_EXP_DECAY2}
   \alpha_{Z}(n)\leq \bar\alpha_{0}\exp\left\{- \balpha_{1} n\right\},\quad n\geq 1,\quad \bar\alpha_{0}>0,\quad \bar\alpha_{1}>0.
\end{eqnarray}
\end{description}
 Then  there are \( \delta'>0 \) and \( \zeta_{0}>0  \), such that
the inequality
\begin{eqnarray}
\label{MINEQ}
\PP\left\{\sqrt{\frac{n}{\log n}}
\left\|\varphi_{n}- \varphi   \right\|_{L_{\infty}(\mathbb{R},w)}>\zeta \right\}&\leq & B \zeta^{-p}n^{-1-\delta' }.
\end{eqnarray}
holds for any \( \zeta >\zeta_{0}  \) and  some positive constant \( B \) not depending on \( \zeta  \) and \(n.\)
\end{prop}
\begin{proof}
Denote 
\begin{eqnarray*}
\mathcal{W}_{n}^{1}(v) &:=&
\frac{w(v) }{n} \;
	\sum_{j=1}^{n} 
	\Bigl(
			e^{(u+\ii v) Z_{j}} 
			I\left\{ e^{u Z_{j }} < \Xi_{n} \right\} 
			-		
			\E \left[	
				e^{(u+\ii v) Z} 
				I\left\{ e^{u Z} < \Xi_{n} \right\} 
			\right]
	\Bigr),\\
\mathcal{W}_{n}^{2}(v) &:=&
\frac{w(v) }{n} \;
	\sum_{j=1}^{n} 
	\Bigl(
			e^{(u+\ii v) Z_{j}} 
			I\left\{ e^{u Z_{j }} \geq \Xi_{n} \right\} 
			-		
			\E \left[	
				e^{(u+\ii v) Z} 
				I\left\{ e^{u Z} \geq \Xi_{n} \right\} 
			\right]
	\Bigr),
\end{eqnarray*}
where \(Z\) is a random variable with stationary distribution of \(Z_{j}\). The main idea of the proof is to show that 
\begin{eqnarray}
\label{aim1}
\PP\left\{
|\mathcal{W}_{n}^{1}(v)|>\zeta 
\sqrt{\frac{\log n}{n}}
\right\}&\leq & \B_{1} \zeta^{-p}n^{-1-\delta' },\\
\label{aim2}
\PP\left\{
|\mathcal{W}_{n}^{2}(v) |>\zeta 
\sqrt{\frac{\log n}{n}}
\right\} &\leq & \B_{2} \zeta^{-p}n^{-1-\delta' },
\end{eqnarray}
with \(\Xi_{n}=...\)  and some positive \(\B_{2}\) and \(\B_{2}\). 

\textbf{Step 1.}  The aim of the first step is to show \eqref{aim1}. The proof follows the same lines as the proof of Proposition A.3 from  \cite{panov2013a}.

\textbf{1.1.}  
Consider the sequence \( A_{k}=e^{k},\, k\in \mathbb{N} \) and cover each
interval \( [-A_{k},A_{k}] \) by
\( M_{k}=\left(\lfloor 2A_{k}/\gamma \rfloor +1  \right) \) disjoint small
intervals \( \Lambda_{k,1},\ldots,\Lambda_{k,M_{k}} \) of
the length \( \gamma. \) Let \( v_{k,1},\ldots, v_{k,M_{k}} \) be the centers
of these intervals. We have for any natural \( K>0 \)
\begin{multline*}
\max_{k=1,\ldots,K}\sup_{A_{k-1}<| v |\leq A_{k}}|\mathcal{W}_{n}^{1}(v)|\leq 
\max_{k=1,\ldots,K}\max_{1\leq m \leq M_{k}}
\sup_{v\in \Lambda_{k,m}}|\mathcal{W}_{n}^{1}(v)-\mathcal{W}_{n}^{1}(v_{k,m})|
\\
+\max_{k=1,\ldots,K}\max_{\Bigl\{\substack{ 1\leq m \leq M_{k}:\\
| v_{k,m} |>A_{k-1}}\Bigr\}}|\mathcal{W}_{n}^{1}(v_{k,m})|.
\end{multline*}
Hence for any positive \(\lambda\),
\begin{multline}
\label{DEC1}
\PP\left( \max_{k=1,\ldots,K}\sup_{A_{k-1}< | v |\leq A_{k}}|\mathcal{W}_{n}^{1}(v)|>\lambda \right)\\
\leq
\PP\left(\sup_{| v_{1}-v_{2} |<\gamma}|\mathcal{W}_{n}^{1}(v_{1})-\mathcal{W}_{n}^{1}(v_{2})|>\lambda/2\right)
\\
+
 \sum_{k=1}^{K}\sum_{\Bigl\{\substack{ 1\leq m \leq M_{k}:
| v_{k,m} |>A_{k-1}}\Bigr\}}\PP(|\mathcal{W}_{n}^{1}(v_{k,m})|>\lambda/2).
\end{multline}
The aim of the next two steps is to get the upper bounds for the summands in the right hand side, where
 \(\lambda\) is taken  in the form \( \lambda=\zeta \sqrt{(\log n) / n} \) with arbitrary large enough \(\zeta\).  
 
\textbf{1.2.}  We proceed with the first summand in \eqref{DEC1}. It holds for any \( v_{1},v_{2}\in \mathbb{R} \)
\begin{eqnarray}
\label{WNDIFF}
\nonumber
|\mathcal{W}_{n}^{1}(v_{1})-\mathcal{W}_{n}^{1}(v_{2})|&\leq&
\left|
	w( v_{1})-w( v_{2} )
\right|
\times
\max_{v}  
\left| 
	\frac{\mathcal{W}_{n}^{1} (v)}{w(v)}
\right|
\\ \nonumber && \cdot\hspace{2cm}
\left| 
	\frac{\mathcal{W}_{n}^{1} (v_{1})}{w(v_{1})}
	-
	\frac{\mathcal{W}_{n}^{1} (v_{2})}{w(v_{2})}
\right|
\times
\max_{v} \left[
	w(v)
\right]
\\
\nonumber
&\leq&
2 \: \Xi_{n}
\left|
	w( v_{1})-w( v_{2} )
\right|\\ \nonumber && \hspace{0.5cm}
+\frac{1}{n}\sum_{j=1}^{n}
\Bigl[
	\left| 
		e^{(u+\ii v_{1}) Z_{j}} - e^{(u+\ii v_{2}) Z_{j}}
	\right| I\left\{ e^{u Z_{j}} < \Xi_{n} \right\} 
\Bigr]\\\nonumber
&& \hspace{0.5cm}+
\Bigl|
\E \left[
		\left(
			 e^{(u+\ii v_{1}) Z} - e^{(u+\ii v_{2}) Z}
	 	\right)
	 I\left\{ 
	 		e^{u Z} < \Xi_{n} 
		\right\} 
	\right]
\Bigr|
\\
&\leq& \left| v_{1}-v_{2} \right|\:\Xi_{n}\:
\left[ 2 \:
L_{w}+\frac{1}{n}\sum_{j=1}^{n}| Z_{j}|+\E| Z | \right],
\end{eqnarray}
where \( L_{\omega } \) is the Lipschitz constant of \( w\) and \(Z\) is a random variable distributed by the stationary law of the sequence \(\left\{Z_{j}\right\}\). Next, the Markov inequality implies
\begin{eqnarray*}
\PP\left\{ \frac{1}{n}\sum_{j=1}^{n}\Bigl[|  Z_{j} |-\E| Z |\Bigr]>
c \right\}\leq c^{-p}n^{-p}\:\E\left| \sum_{j=1}^{n}\Bigl[| Z_{j} |-\E| Z | \Bigr] \right|^{p}
\end{eqnarray*}
for any \( c >0. \)
Using now Yokoyama   inequality \cite{Yoko} and taking into account the assumptions of the continuity of moments of \(Z_{j}\) and the assumption 1 from Lemma \ref{EIB},  we get
\begin{eqnarray*}
 \E\left| \sum_{j=1}^{n}\Bigl[
 	|  Z_{j} |-\E| Z |
\Bigr] \right|^{p}\leq C_{p}(\bar\alpha)n^{p/2},
\end{eqnarray*}
where \( C_{p}(\bar\alpha) \) is some constant depending on \(\bar\alpha=(\bar\alpha_{0},\bar\alpha_{1})\) and \( p\).
Returning to our choice of  \(\gamma\) and \(\lambda\), which in particularly yields that \( \gamma=\lambda / \zeta =\sqrt{(\log n) / n}, \) 
we obtain from \eqref{WNDIFF}
\begin{multline*}
\label{LINEQ}
\PP\Bigl\{
	\sup_{| v_{1}-v_{2}|<\gamma}|\mathcal{W}_{n}^{1}(v_{1})-\mathcal{W}_{n}^{1}(v_{2})|>\lambda/2\Bigr\}
	\leq
	\\
 \PP\left\{ 
 	\frac{1}{n}\sum_{j=1}^{n}
	\Bigl[
		|  Z_{j} |-\E| Z |
	\Bigr]>\frac{\zeta}{2 \Xi_{n}}-
 2 L_{w}-2\E|Z| 
\right\}\\
\leq 
B_{0} \: c_{p}(\balpha) \Bigl(
	\zeta/\left(2 \Xi_{n}\right)- 
 2 L_{w}-2\E|Z| 
\Bigr)^{-p} n^{-p/2}
\leq
B_{1}\zeta^{-p} \:\Xi_{n}^{p} \:n^{-p/2}
\end{multline*}
with some constants \( B_{0}, B_{1} \) not depending on \( \zeta \) and \( n, \) provided \(\zeta\) is large enough.

\textbf{1.3.}  
Now we turn to the second term on the right-hand side of  \eqref{DEC1}.
Applying  Lemma~\ref{EIB} with \(G_{k}=n \Ree\left[ \mathcal{W}_{n}^{1}(u_{k,m})\right]\) and \(\beta = n \lambda\), we get
\begin{eqnarray*}
\PP\left(|\Ree\left[ \mathcal{W}_{n}^{1}(v_{k,m}) \right]|>\lambda/4\right) \leq \K, 
\end{eqnarray*}
where 
\begin{eqnarray*}
\K :=  \exp\left(
-\frac{
	B_{3} \lambda^{2}n
}{
	B_{2} \Xi_{n}^{2} w^{2}(A_{k-1})\log^{2(1+\varepsilon) }(\Xi_{n} w(A_{k-1}))+\lambda\log^{2}(n) \Xi_{n}
w(A_{k-1})
}
\right)
\end{eqnarray*}
with some constants \( B_{2} \) and  \( B_{3} \) depending only on the characteristics
of the process \( Z \).  Similarly, applying the same result with \(G_{k}=n \Imm \left[ \mathcal{W}_{n}^{1}(u_{k,m})\right]\), we conclude that
\begin{eqnarray*}
\PP\left(|\Imm\left[ \mathcal{W}_{n}^{1}(v_{k,m}) \right]|>\lambda/4\right)
\leq \K,
\end{eqnarray*}
and therefore 
\begin{eqnarray*}
\sum_{\{ | v_{k,m} |>A_{k-1} \}}\PP(|\mathcal{W}_{n}^{1}(v_{k,m})|>\lambda/2)\leq
\left(\lfloor 2A_{k}/\gamma \rfloor +1  \right)\K.
\end{eqnarray*}
Set now \(\gamma= \sqrt{(\log n)/n}\) and \(\lambda = \zeta \sqrt{(\log n)/n}\) and note that under our choice of \(\Xi_{n}\), 
\[
	\Xi_{n}^{2} w^{2}(A_{k-1})\log^{2(1+\varepsilon) }(\Xi_{n} w(A_{k-1})) \gtrsim \lambda\log^{2}(n) \Xi_{n}
w(A_{k-1}).
\]
Therefore,
\begin{multline*}
\sum_{\{ | v_{k,m} |>A_{k-1} \}}\PP(|\mathcal{W}_{n}^{1}(v_{k,m})|>\lambda/2)
\\
\lesssim   A_{k}
\sqrt{\frac{n}{\log(n) }}
\exp\left(-\frac{B\zeta^{2}\log (n)}{w^{2}(A_{k-1})\Xi_{n}^{2}\log^{2(1+\varepsilon)}(w(A_{k-1}))}\right)  , \quad n\to \infty
\end{multline*}
with some constant \( B>0. \)
Fix \( \theta>0 \) such that \( B\theta>1 \)  and compute
\begin{multline*}
 \sum_{\{ | v_{k,m} |>A_{k-1} \}}\PP(|\mathcal{W}_{n}^{1}(v_{k,m})|>\lambda/2) \\ 
 \lesssim 
\sqrt{\frac{n}{\log(n) }}
\exp\Bigl\{k-\theta B(k-1) -B(k-1)(\zeta^{2}(\log n/\Xi_{n})-\theta)\Bigr\} 
 \\
\lesssim
 \sqrt{\frac{n}{\log(n) }}
 e^{k(1-\theta B)}  e^{-B(k-1)(\zeta^{2}(\log n/\Xi_{n})-\theta)}.
\end{multline*}
Since \( \zeta^{2}(\log n/\Xi_{n})>\theta \), we arrive at 
\begin{multline*}
 \sum_{k=2}^{K}\sum_{\{ |v_{k,m} |>A_{k-1} \}}\PP(|\mathcal{W}_{n}(v_{k,m})|>\lambda/2)
\\ \lesssim     
 \sqrt{\frac{n}{\log(n) }}
 e^{-B(\zeta^{2}(\log n/\Xi_{n})-\theta)}
 \Bigl[
	 \sum_{k=2}^{K} e^{k(1-\theta B)}
\Bigr]\\
  \lesssim 
\log^{-1/2} (n) \exp\Bigl\{-B\zeta^{2}(\log n/\Xi_{n}) + \log(n) \Bigr\}.
\end{multline*}

Taking large enough \( \zeta >0 \), we get \eqref{aim1}.

\textbf{Step 2}. Now we are concentrated on \eqref{aim2}. The idea of the proof given below was published in \cite{Belpreprint}, Proposition 7.4.

Consider the sequence
\[
R_{n}(v) := 
\frac{1 }{n} \;
	\sum_{j=1}^{n} 
			e^{(u+\ii v) Z_{j}} 
			I\left\{ e^{u Z_{j }} \geq \Xi_{n} \right\}. 
\]
By the Markov inequality we get 
\begin{multline*}
	\left|
		\E \left[
			R_{n}(u)
		\right]
	\right|
	\leq \E \left[
		e^{u Z_{j}}
	\right]
	\PP \left\{ e^{u Z_{j }} \geq \Xi_{n} \right\} 
	\\ \leq 
	\Xi_{n}^{-p}
	\;\;
	\E \left[
		e^{u Z_{j}}
	\right]
	\;
	\E
	\left[
		e^{u p Z_{j}}
	\right] =  o\Bigl(\sqrt{(\log n)/ n}\Bigr) 
\end{multline*}
Set \(\nu_{k}=2^{k}, k \in 1,2,...\), then it holds 
\begin{multline*}
\sum_{k=1}^{\infty}\PP\Bigl\{ \max_{j=1..\eta_{k+1}}   e^{u Z_{j}} \geq \Xi_{\eta_{k}}\Bigr\}
\leq
	\sum_{k=1}^{\infty} \eta_{k+1} \PP\{ e^{u Z} \geq \Xi_{\eta_{k}}\} \\
	\leq 
	\E e^{p u Z} 
	\sum_{k=1}^{\infty} \eta_{k+1} \Xi_{\eta_{k}}^{-p}<\infty.
\end{multline*}
By the Borel-Cantelli lemma,
\[
	\PP\Bigl\{ \max_{j=1..\eta_{k+1}}   e^{u Z_{j}} \geq \Xi_{\eta_{k}} \quad\mbox {for infinitely many k}\Bigr\} = 0.
\]
From here it follows that \(R_{n}(u)- \E R_{n}(u) = o\Bigl(\sqrt{(\log n)/ n}\Bigr)\). This completes the proof. 
\end{proof}
\begin{lem} 
\label{lemlem}
Let the measure \(\bar{\nu}\) be such that \(\|\bar{\nu}^{(r)}\|_{\infty} \leq C_{1}\)  for some positive \(C_{1}\), the weighting function \(w_{n}\) admits the property \(w_{n}=V_{n}^{-k} w(v/V_{n})\) for some \(k>0\) and function \(w\) satisfying 
\[
	\|
		\F_{w(u) /u^{r}} (\cdot)
	\|_{L_{1}} \leq C_{2}
\]
with some \(C_{2}>0\). Then 
\[
\Bigl|
	 \int_{0}^{\infty} w_{n}(v) 
 		\F_{\bar{\nu}} (v) dv
\Bigr| \lesssim V_{n}^{-(r+k)}, \qquad n \to \infty
.\]
\end{lem}
\begin{proof}
 Following \cite{BelReiss}, we apply the Plancherel identity:
\begin{eqnarray*}
\label{cc}
\Bigl|
	 \int_{0}^{\infty} w_{n}(v) 
 		\F_{\bar{\nu}} (v) dv
\Bigr|  &=& 2 \pi \left| 
		\int_{\R} \bar\nu^{(r)} (x) 
		\overline{
		\F^{-1}_{w_{n}(\cdot)/(\ii \cdot)^{r}} (x) 
		}
		dx
	\right|\\
	&\leq&
	2 \pi \:
	V_{n}^{-(r+k)} \|\bar{\nu}^{(r)}\|_{\infty}\;
	\|
		\F_{w(u) /u^{r}} (\cdot)
	\|_{L_{1}} 	\lesssim
	V_{n}^{-(r+k)}.
\end{eqnarray*}
\end{proof}

\begin{lem}[analogue of the Parseval-Plancherel theorem for Mellin transform]
\label{lemmelin}
Let \(X_{1}\) and \(X_{0}\) be two L{\'e}vy process with expontional functionals that have densities \(p_{0}\) and \(p_{1}\), and Mellin transforms \(M_{0}\) and \(M_{1}\), resp. For any \(b \in \R\), it holds 
\begin{multline*}
	\int_{0}^{\infty} x^{b} \left(
		p_{0}(x) - p_{1}(x)
	\right)^{2}
	dx \\
	=
	(2 \pi)^{-1/2}
	\int_{-\infty}^{\infty} \left|
		M_{0} (b/2 +1/2+\ii v) - 		M_{1} (b/2 +1/2+\ii v) 
	\right|^{2}
	dv.
\end{multline*}
\end{lem}

\bibliographystyle{plain}
\bibliography{Panov_bibliography}

\end{document}